\documentclass[JHEP,12pt]{article}
\usepackage{jheppub}
\usepackage{amssymb,amsmath,amsthm}
\usepackage{mathrsfs}  
\usepackage[shortlabels]{enumitem}
\usepackage{braket}
\usepackage{breqn}
\usepackage{comment} 
\usepackage{subcaption}
\usepackage{graphicx, xcolor, varwidth}
\usepackage{color}
\usepackage{float}

\newtheorem{thm}{Theorem}
\theoremstyle{definition}
\newtheorem{conj}[thm]{Conjecture}
\newtheorem{defn}[thm]{Definition}
\newtheorem{cor}[thm]{Corollary} 
\newtheorem{lem}[thm]{Lemma}
\theoremstyle{remark}
\newtheorem{rem}[thm]{Remark}
\newtheorem{conv}[thm]{Convention}

\DeclareMathOperator{\cl }{cl}
\DeclareMathOperator{\setint }{int}

\DeclareMathOperator{\A}{Area}

\DeclareMathOperator{\minEW}{minEW}
\DeclareMathOperator{\maxEW}{maxEW}
\DeclareMathOperator{\EW}{EW}
\DeclareMathOperator{\CW}{CW}
\renewcommand{\S}{S_{\rm gen}}
\newcommand{\emax}{e_{\rm max}}
\newcommand{\emin}{e_{\rm min}}
\newcommand{\smax}{\Sigma_{\rm max}}
\newcommand{\smin}{\Sigma_{\rm min}'}

\newcommand{\scri}{\mathscr{I}}

\newcommand{\hmingen}{H_{\rm min,gen}}
\newcommand{\hmg}{H_{\rm max,gen}}

\newcommand{\hmax}{H_{\rm max}}

\title{Fundamental Complement of a Gravitating Region}
\author{Raphael Bousso and Sami Kaya}

\affiliation{Center for Theoretical Physics and Department of Physics,\\
University of California, Berkeley, California 94720, U.S.A. 
} 
\emailAdd{bousso@berkeley.edu}
\emailAdd{samikaya@berkeley.edu}

\abstract{
Any gravitating region $a$ in any spacetime gives rise to a generalized entanglement wedge, the hologram $e(a)$. Holograms exhibit properties expected of fundamental operator algebras, such as strong subadditivity, nesting, and no-cloning.  But the entanglement wedge $\EW$ of an AdS boundary region $B$ with commutant $\bar B$ satisfies an additional condition, complementarity: $\EW(B)$ is the spacelike complement of $\EW(\bar B)$ in the bulk.

Here we identify an analogue of the boundary commutant $\bar B$ in general spacetimes: given a gravitating region $a$, its \emph{fundamental complement} $\tilde a$ is the smallest wedge that contains all infinite world lines contained in the spacelike complement $a'$ of $a$. We refine the definition of $e(a)$ by requiring that it be spacelike to $\tilde a$. We prove that $e(a)$ is the spacelike complement of $e(\tilde a)$ when the latter is computed in $a'$.

We exhibit many examples of $\tilde a$ and of $e(a)$ in de Sitter, flat, and cosmological spacetimes. We find that a Big Bang cosmology (spatially closed or not) is trivially reconstructible: the whole universe is the entanglement wedge of any wedge inside it. But de Sitter space is not trivially reconstructible, despite being closed. We recover the AdS/CFT prescription by proving that $\EW(B)=e($causal wedge of $B$).}

\makeatletter
\gdef\@fpheader{\mbox{}}
\makeatother

\begin{document}
\maketitle
\section{Introduction}
The past few decades have witnessed a remarkable convergence between gravity, quantum field theory, and information theory, culminating in the development of the holographic principle~\cite{Bousso:2002ju}—a radical proposal that the information content of a region of spacetime is encoded not within its volume, but rather on its boundary surface. This idea was originally motivated by the Bekenstein-Hawking entropy formula for black holes~\cite{Bekenstein:1972tm}, and later refined and made concrete through the Bousso bound~\cite{Bousso:1999xy}, which generalizes the concept of entropy bounds to arbitrary null hypersurfaces in spacetime. 

The holographic principle thus suggests a profound reduction in the number of degrees of freedom needed to describe a gravitational system, pointing toward a fundamentally nonlocal and lower-dimensional structure underlying quantum gravity. A concrete realization is provided by the Anti-de Sitter/Conformal Field Theory (AdS/CFT) correspondence~\cite{Maldacena:1997re}, which posits quantum gravity in asymptotically AdS spacetime is a CFT  defined on its conformal boundary. 

A particularly rich area of research within AdS/CFT is the study of subregion duality~\cite{Bousso:2012sj,Czech:2012bh,Bousso:2012mh}. A given subregion $B$ of the boundary CFT encodes information about a specific region EW($B$) in the bulk, the entanglement wedge of $B$. This correspondence has far-reaching implications for our understanding of bulk locality, quantum error correction, and the emergence of spacetime geometry from patterns of entanglement. A further refinement\footnote{Readers unfamiliar with the max-/min-refinement can safely ignore it by replacing $\maxEW$ and $\minEW$ by $\EW$, and replacing $\emax$ and $\emin$ by $e$, throughout this paper, except in Fig.~\ref{fig:AdS_2sided} (left).} \cite{Akers:2020pmf,Akers:2023fqr} distinguishes between the max-entanglement wedge --- roughly, the region that can be reconstructed with high probability --- and the min-entanglement wedge, which can be reconstructed with small probability. This distinction arose first in the quantum communication task of state merging~\cite{RenWol04a}. It becomes important when the bulk contains incompressible quantum states, for which the von Neumann entropy is poorly suited. The entanglement wedges of AdS boundary regions obey the complementarity relation
\begin{equation}\label{eq:oldcomp}
    \minEW(B)' = \maxEW(\bar B)~.
\end{equation}
That is, the largest bulk region that cannot be reconstructed from $B$ even with small probability is precisely the bulk region that can be reconstructed from $\bar B$ with high probability.

The concept of entanglement wedges was recently extended to arbitrary spacetimes. The role of the CFT region $B$ is played by any bulk wedge\footnote{A wedge is a spacetime region $a$ that equals its double spacelike complement $a''$; see Definitions~\ref{def:sc} and \ref{def:covwedge}. Equivalently, a wedge is the full causal development of a partial Cauchy slice. See Fig.~\ref{fig:union}.} $a$. A specific prescription then identifies a max- and a min-entanglement wedge to $a$~\cite{Bousso:2022hlz,Bousso:2023sya}. 
An interesting feature of the Bousso–Penington proposal is that a distinction between the max- and min-holograms, $\emax(a)$ and $\emin(a)$, arises even at the classical level.

\begin{figure}[t]
\begin{center}
  \includegraphics[width=0.8\linewidth]{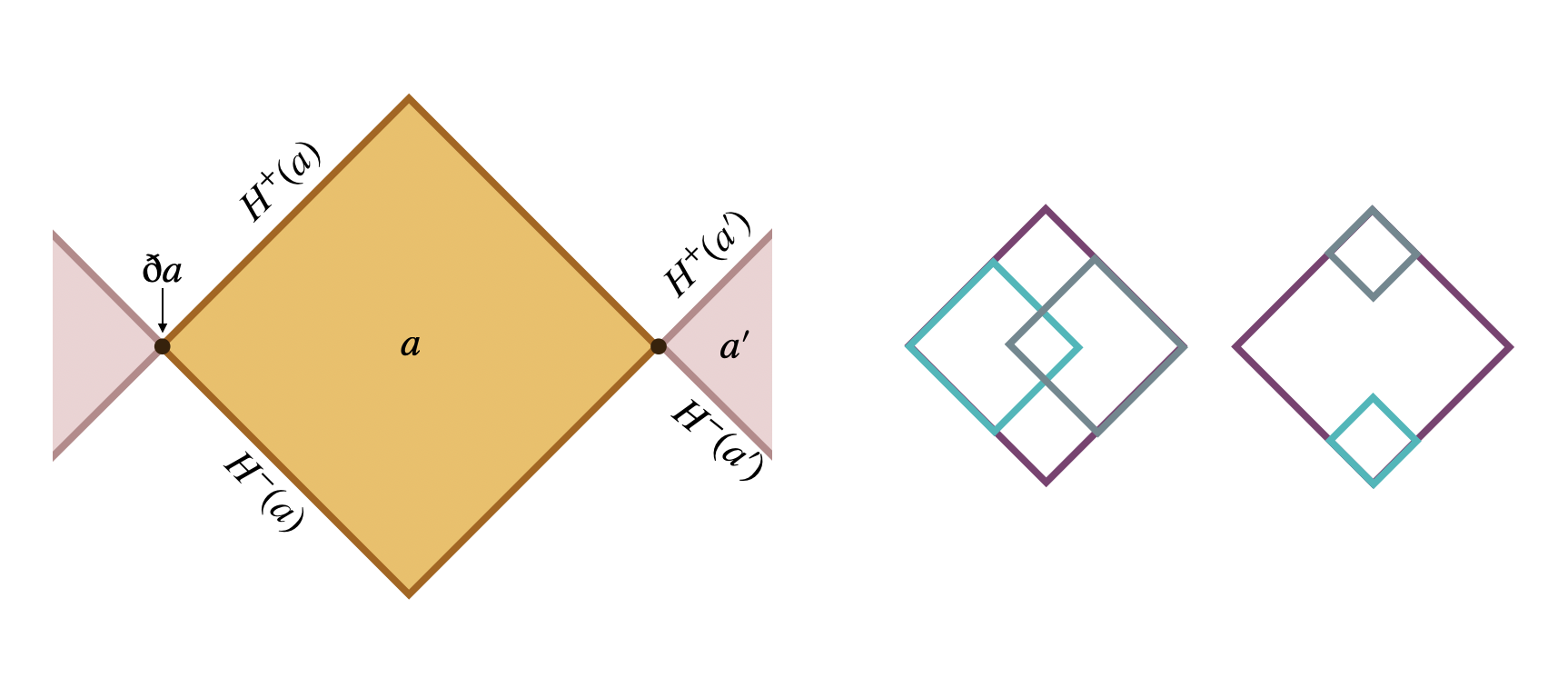}
\end{center} \vspace{-5ex}
\caption{Wedges in 1+1 dimensions. \emph{Left:} A wedge $a$ and its complement wedge $a'$, their common edge $\eth a=\eth a'$, and their past and future Cauchy horizons.
\emph{Right:} The wedge union (purple) of two wedges (turquoise, blue) is the smallest wedge that contains both.}
\label{fig:union}
\end{figure}

These generalized entanglement wedges, or \emph{holograms}, constrain how gravity processes information intrinsically within the bulk, without invoking a boundary dual. The proposal thus marks a first step toward a kind of gravity–gravity duality. The precise nature of this duality remains an open question. 

A possible interpretation is that semiclassical operators in a given input region $a$ generate the full algebra of the fundamental theory associated with its entanglement wedge. Evidence for this interpretation comes from the fact that holograms nontrivially obey strong subadditivity, nesting, and no-cloning properties that are indicative of a relation to a fundamental quantum algebra~\cite{Bousso:2023sya,Bousso:2024iry}. 

In this paper, we will show that holograms also obey a generalization of the complementarity relation \eqref{eq:oldcomp}. It is not obvious \emph{a priori} that such a generalization exists. A naive attempt would be to replace  $\minEW\to \emin$, $\maxEW\to \emax$, $B\to a$, and $\bar B \to a'$, in Eq.~\eqref{eq:oldcomp}, where $a'$ is the spacelike complement of $a$ in the arbitrary gravitating spacetime $M$. This would yield a false statement. For example, if $a$ is a thin shell in flat space, then $\emin(a)$ includes its interior, but so does $\emax(a')$. 

The problem is that the complement wedge $a'$ is obviously not \emph{fundamentally} independent of $a$, since a part of $a'$ lies inside of $\emin(a)$ and hence can be reconstructed from $a$. Clearly, therefore, the spacelike complement $a'$ is not an appropriate analogue of $\bar B$, the boundary commutant of the CFT region $B$. Perhaps the most surprising result in this paper is that a more appropriate analogue --- the \emph{fundamental} complement of a bulk region --- can be defined at all.

\paragraph{Summary and Outline} In Section~\ref{sec:def}, we refine the Bousso–Penington proposal by introducing the fundamental complement $\tilde a$ of any wedge $a$. Our approach can be motivated as follows. In AdS/CFT, $\maxEW(B)$ (and similarly $\minEW(B)$) is constrained by the condition that its conformal boundary at infinity must be $B$. This homology condition can be broken up into two separate conditions: the closure of $\maxEW(B)$ must contain $B$ and must not overlap with $\bar B$. For a bulk input region $a$, the analogous two conditions are that the entanglement wedge $\emax(a)$ must contain $a$, and that $\emax(a)$ must not overlap with the past or future of $\tilde a$. In short, the two conditions become the double inclusion requirement\footnote{Throughout the paper we evaluate the tilde before the prime: $\tilde a'\equiv (\tilde a)'$.} $a\subset \emax(a)\subset \tilde a'$ (and similarly for $\emin$).

In Sec.~\ref{sec:fcdef}, we define $\tilde a$ as the smallest wedge that contains all timelike curves that never enter the past or future of $a$ and have infinite proper duration in both the past- and future direction. Thus, $\tilde a$ can be thought of as a direct generalization, to arbitrary spacetimes, of the causal wedge of the boundary complement $\bar B$ in AdS/CFT. Note that the causal wedge of $\bar B$ is spacelike to $\minEW(B)$ and $\maxEW(B)$ by causal wedge inclusion~\cite{Wall:2012uf}. 

Physically, it is clear in arbitrary spacetimes that operators in $\tilde a$ cannot be holographically reconstructed from $a$, simply because all points in $\tilde a$ can be directly causally manipulated from conformal infinity. Hence, it is safe to exclude $\tilde a$ when we construct the entanglement wedge of $a$. In Sec.~\ref{sec:accessdef} we re-develop the notion of holographic accessibility from $a$ in a form that incorporates this requirement.


\begin{figure}
    \centering
    \includegraphics[width=1\linewidth]{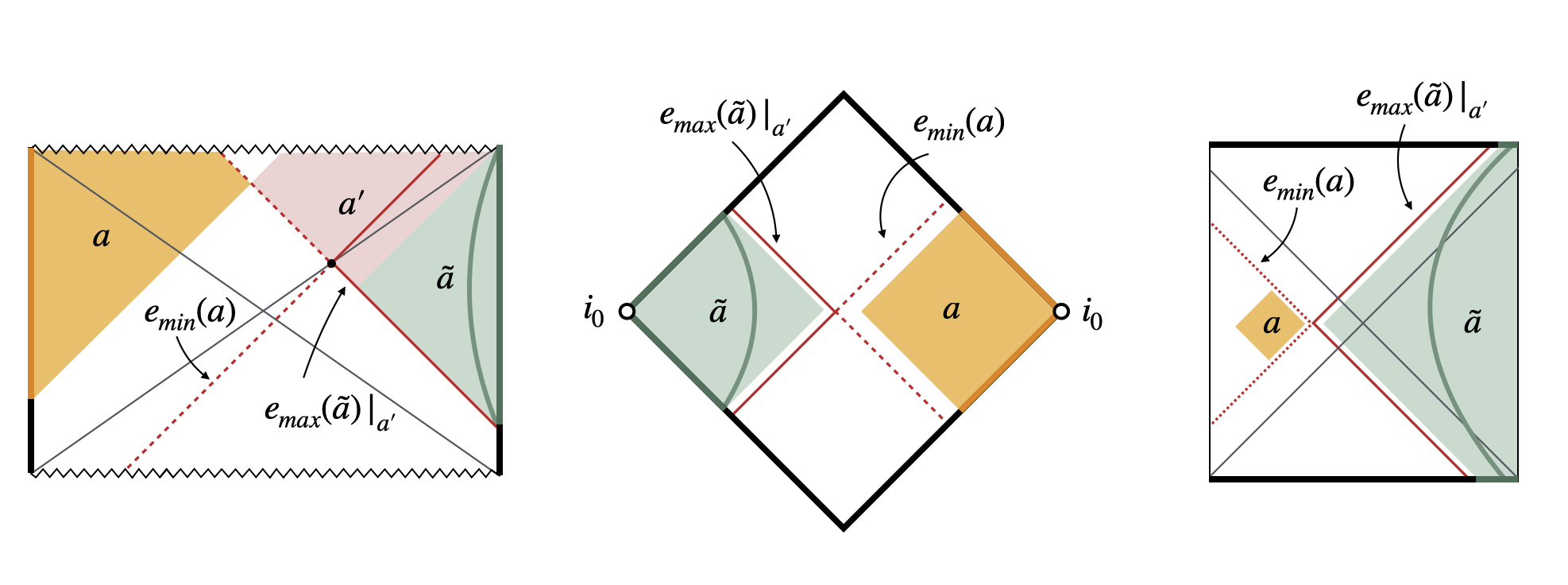}
    \caption{\emph{Left:} Two-sided AdS black hole with some matter present, which makes the apparent horizons (grey X) spacelike. The min-hologram $\emin(a)$ of the wedge $a$ is (roughly) the smallest wedge that includes $a$ and has positive null expansions (red solid lines). The fundamental complement $\tilde a$ (green shaded) of $a$ is the smallest wedge that contains all causal curves (such as the thick green curves) that are both past- and future-infinite in the spacelike complement $a'$ of $a$. The max-hologram of $\tilde a$ in $a'$, $\emax(\tilde a)|_{a'}$ is the largest wedge accessible from $\tilde a$ (red dashed lines); it is the spacelike complement of $\emin(a)$.  \emph{Middle and Right:} Examples of the fundamental complement, and of entanglement wedge complementarity, in asymptotically flat and asymptotically de Sitter space; see Sec.~\ref{sec:examples} for a detailed discussion. Note that $\emin(a)$ is not the entire spacetime even though de Sitter space is spatially closed.}  
    \label{fig:AdS_2sided}
\end{figure}

In Sec.~\ref{sec:edef}, we modify the Bousso–Penington proposal. As before, $\emax(a)$ as the union of all wedges that are accessible from $a$, but with our new notion of accessibility replacing that of Ref.~\cite{Bousso:2023sya}. Our choice is motivated by the physical considerations described in the previous paragraphs, and further supported by the complementarity theorem it allows us to prove.\footnote{\emph{Some} modification of Ref.~\cite{Bousso:2023sya} would have been required in any case, due to shortcomings of the original formulation. The minimal entanglement wedge $\emin(a)$ was defined as the intersection of certain sets whose conformal edge~\cite{Bousso:2023sya} coincides with that of $a$. This requirement, while motivated by the homology constraint on entanglement wedges in AdS/CFT, is actually inconsistent with the proof provided in~\cite{Bousso:2023sya} that $\emin$ exists. One could drop the requirement that the conformal edges agree, but this would lead to $\emin(a)=M$ in all cases, which is physically unreasonable. Separately, a problem arises for the max-entanglement wedge of a Rindler wedge in asymptotically flat space, which would also be the entire spacetime with the prescription in Ref.~\cite{Bousso:2023sya}.}

In Section~\ref{sec:prop}, we establish key properties of the fundamental complement, and we prove that our new definitions of $\emax$ and $\emin$ satisfy a complementarity theorem. 

In Sec.~\ref{sec:fcprop}, we show that $\tilde a$ is the exterior of a past and a future causal horizon and hence is antinormal. (In the classical limit, antinormal means that the null expansions orthogonal to its edge are nonpositive.) We also show that the fundamental complement of $\tilde a$ in the spacetime $a'$ vanishes. In Sec.~\ref{sec:comp}, we prove that the min-hologram of $a$ is complementary to the max-hologram of $\tilde a$ computed in the spacetime $a'$: 
\begin{equation}
    \emin(a)' = \emax(\tilde a)|_{a'} ~.
\end{equation}

In Section~\ref{sec:examples} we present explicit constructions of the fundamental complement $\tilde a$ in various kinds of spacetimes with varying asymptotic structure. 

We begin by noting, in Sec.~\ref{sec:cosmo} that $\tilde a=\varnothing$ in any Big Bang (or Big Crunch) cosmology. Such spacetimes are trivially reconstructible: $\emax(a)=M$ for any $a\subset M$. Interestingly, this is independent of the spatial topology. It requires only the absence of timelike curves that the are both past- and future-infinite. 

In Sec.~\ref{sec:flat} we discuss several interesting examples in Minkowski space. Here $\tilde a$ plays an important role in preventing a Rindler wedge from accessing the whole spacetime. We are also able to compute, for the first time, the entanglement wedge of a Rindler wedge with a local protrusion. 

In Sec.~\ref{sec:dS} we turn to asymptotically de Sitter spacetimes. Again, $\tilde a$ plays a nontrivial role. Even though de Sitter is spatially closed, we find that is not trivially reconstructible, because $\tilde a$ need not vanish.

In Section~\ref{sec:ads}, we recover the entanglement wedge prescriptions for AdS/CFT as a special case of our prescription. We show that 
\begin{equation}\label{eq:bb}
        \maxEW(B) = \emax[\CW(B)]~,~~~\minEW(B) = \emin[\CW(B)]~,
\end{equation}
where $\CW(B)$ is the causal wedge of the boundary region $B$ in the bulk.

In fact, the entanglement wedge of any subset of conformal infinity can defined by Eq.~\eqref{eq:bb}, regardless of asymptotic structure. For example, the entanglement wedge of two semi-infinite null generators of $\scri^\pm$ in asymptotically flat space is simply the entanglement wedge of a Rindler wedge, see Def.~\ref{def:rindler} in Sec.~\ref{sec:flat}. The entanglement wedge of timelike separated points $p\in \scri^+$, $q\in \scri^-$ in de Sitter space is the entanglement wedge of a causal diamond~\cite{Bousso:2000nf} and thus~\cite{Gao:2000ga} the whole spacetime.

That said, we would like to de-emphasize the idea that holographic reconstruction should originate with a conformal boundary. On the contrary, the perspective developed here and in Ref.~\cite{Bousso:2022hlz,Bousso:2023sya,Bousso:2024iry}, is that holography is fundamentally the reconstruction of a larger bulk region from (the full algebra generated by the semiclassical operators in) the smaller one.

In Appendix~\ref{app:oldthms} we show that our improved definitions of $\emax$ and $\emin$ satisfy key properties established (or claimed) in Ref.~\cite{Bousso:2023sya} (for the versions of $\emax$ and $\emin$ defined there). We provide a valid proof that $\emin(a)\supset a$. We prove nesting of $\emin$, a no-cloning theorem, and the inclusion of $\emax$ in $\emin$. We also prove the strong subadditivity of the generalized entropy of holograms, in the strong form established in Ref.~\cite{Bousso:2024iry}. In Appendix~\ref{app:largecomp} we describe a possible alternative definition of the fundamental complement, which to us appears less well motivated.

\emph{Note added:} as this manuscript was finalized, we became aware of Ref.~\cite{Gupta:2025jlq}, which associates entanglement wedges to non-AdS boundary regions (but not to bulk regions). Our proposal is more general, but this limit can be captured using Eq.~\ref{eq:bb}. We find that neither the ``orthodox'' nor the ``heterodox'' proposal of Ref.~\cite{Gupta:2025jlq} emerge from ours. However, the ``heterodox'' proposal appears to be equivalent to the modification of our proposal discussed in Appendix~\ref{app:oldthms}, when the latter is restricted to boundary input regions and evaluated via Eq.~\ref{eq:bb}. We have not yet attempted to prove this relation.

\section{Definitions}\label{sec:def}

Let $M$ be a globally hyperbolic Lorentzian spacetime with metric $g$. The chronological and causal future and past, $I^\pm$ and $J^\pm$,
the future and past domains of dependence and 
the future and past Cauchy horizons, 
$D^\pm$ and 
$H^\pm$, 
are defined as in Wald~\cite{Wald:1984rg}. In particular, the definitions are such that $p\notin I^+(\set{p})$ and $p\in J^+(\set{p})$.  We use $\subsetneq$ to denote a proper subset; $\subset$ includes equality. Intersections are taken before unions: $s\cap t\cup u\equiv (s\cap t)\cup u$. Given any set $s\subset M$, $\partial s$ denotes the boundary of $s$ in $M$, and $\cl s\equiv s\cup \partial s$ denotes the closure. If an operation is performed in a spacetime other than $M$, we indicate this through a subscript. (We will highlight this again below.)

\subsection{Wedges, Null Infinity, and the Fundamental Complement}
\label{sec:fcdef}

This subsection establishes definitions that use only the causal structure of $M$. Definitions~\ref{def:cw} and \ref{def:fc} are new and critical to this paper: the causal wedge inside a wedge, and the fundamental complement of a wedge.

\begin{defn}\label{def:sc}
The spacelike complement of a set $s\subset M$ is defined by 
\begin{equation}
    s'\equiv M\setminus \cl I(s) ~,
\end{equation}
where
\begin{equation}
    I(s)\equiv I^+(s)\cup I^-(s)~.
\end{equation}
\end{defn}

\begin{defn}\label{def:covwedge}
A {\em wedge} is a set $a\subset  M$ that satisfies $a=a''$ 
\end{defn}

\begin{defn}
The {\em edge} of a wedge $a$ is defined by  $\eth a \equiv \partial a \backslash I(a)$. The \emph{past and future Cauchy horizon} are $H^\pm(a)\equiv \partial a\cap I^\pm(a)$, and we define $H(a) \equiv H^+(a)\cup H^-(a).$
\end{defn}

\begin{rem}\label{wilem}
The intersection of two wedges $a,b$ is easily shown to be a wedge: $(a\cap b)'' = a\cap b$. Similarly, the \emph{complement wedge} $a'$ is a wedge: $a'''=a'$.
\end{rem}

\begin{defn}
The {\em wedge union} of wedges $a$ and $b$ is $a\Cup b\equiv (a'\cap b')'$; see Fig.~\ref{fig:union}. 
\end{defn}

\
\begin{defn}\label{def:cw}
    Given a wedge $b$, let $C\subset b$ be the set of points that lie on timelike curves that are entirely contained in $b$ and have infinite proper time duration to the past and to the future. We define the \emph{causal wedge in $b$} as
    \begin{equation}
        b_C\equiv C''
    \end{equation}
    Note that $b_C\subset b$ and that $b_C$ is a wedge.
\end{defn}

\begin{defn}[Fundamental Complement]\label{def:fc}
    Given a wedge $a$, we define its \emph{fundamental complement} $\tilde a$ as the causal wedge in its spacelike complement $a'$:
    \begin{equation}
        \tilde a \equiv (a')_C~.
    \end{equation}
\end{defn}

\subsection{Discrete Nonexpansion and Accessibility}
\label{sec:accessdef}

This subsection establishes definitions involving the expansion of lightrays. We follow Refs.~\cite{Bousso:2024iry,Bousso:2025xyc}, except that Def.~\ref{def:accessible} will modify the definition of accessibility~\cite{Bousso:2023sya,Akers:2023fqr} by requiring inclusion in $\tilde a'$. 

We adopt a simplified treatment for clarity and concision. We omit the smoothing parameter $\epsilon$ when discussing conditional max- and min-entropies. We assume that the area and matter contributions to generalized entropies are separable, and we neglect terms of order $G$ and higher. For precise definitions of smoothing and generalized max-entropy, see Ref.~\cite{Akers:2023fqr}. For a detailed discussion of the $G\hbar$ expansion, see Refs.~\cite{Shahbazi-Moghaddam:2022hbw, Bousso:2024iry}.

Readers unfamiliar with max- and min-entropies may easily eliminate these concepts without missing the main points of our work. Throughout the paper, conditional max- and min-entropies can be replaced for simplicity by generalized entropy differences, using Eq.~\eqref{eq:approx}, or even just area differences, using Eq.~\eqref{eq:approxx}. This does not trivialize our main results, such as Def.~\ref{def:fc} of the fundamental complement and the complementarity theorem~\ref{thm:maxmincomp}.

\begin{defn}
Given nested wedges $a\subset b$, the \emph{generalized smooth max-entropy} of $b$ conditioned on $a$ is given by
\begin{equation}\label{eq:hmgdef}
    \hmg(b|a) \approx \left[ \frac{\A(b)-\A(a)}{4G} + \hmax(b|a) + O(G)\right]~.
\end{equation}
Here $\A(a)$ is the area of $\eth a$, and $\hmax(b|a)$ is the (ordinary) smooth max-entropy of the quantum fields in $b$ conditioned on $a$~\cite{RenWol04a}. For compressible quantum states (which are often considered), this becomes a difference of von Neumann entropies:
\begin{equation}
    \hmax(b|a) \approx S(b|a) \equiv S(b)-S(a)~.
\end{equation}
In this approximation,
\begin{equation}\label{eq:approx}
    \hmg(b|a) \approx \S(b|a) \equiv \S(b)-\S(a)~,
\end{equation}
where 
\begin{equation}
    \S(a)\equiv \frac{\A(a)}{4G} + S(a) + O(G)
\end{equation}
is the generalized entropy~\cite{Bekenstein:1972tm} of the wedge $a$.
In many interesting settings, we will be able to neglect the contribution of bulk quantum fields entirely and keep only the areas of the edges:
\begin{equation}\label{eq:approxx}
    \hmg(b|a) \approx \S(b|a) \equiv \frac{\A(b)-\A(a)}{4G}~.
\end{equation}
\end{defn}

\begin{lem}[Strong Subadditivity of $\hmg$]\label{lem:ssa}
    Let $a\subset b$ be wedges. Then for any wedge $c$ spacelike to $b$,
    \begin{equation}\label{eq:ssa}
        \hmg(b\Cup c|a\Cup c)\leq \hmg(b|a)~.
    \end{equation}
\end{lem}
\begin{proof}
     Strong subadditivity holds trivially for the area terms, and as a mathematical theorem for the conditional max entropy of the quantum fields~\cite{Bousso:2023sya}:
    \begin{align}
    \A(b\Cup c)-\A(a\Cup c) & \leq \A(b)-\A(a) ~;\\
    \hmax(b\Cup c|a\Cup c) & \leq \hmax(b|a)~.
    \end{align}
    At higher orders in $G$, strong subadditivity is a conjecture~\cite{Bousso:2024iry}.
\end{proof}

\begin{defn}[Nonexpansion]\label{def:fne}
A wedge $a$ is said to be \emph{future-nonexpanding} at $p\in\eth a$ if there exists an open set $O$ containing $p$ such that~\cite{Bousso:2024iry} 
    \begin{align}\label{eq:futnonexp}
        \hmg(b|a) \leq 0 \text{~for~all~wedges~}b\supset a\text{~such~that~} 
        \eth b \subset \eth a \cup [H^+(a')\cap O]~.
    \end{align}
This definition is ``discrete'' in the sense that no numerical value is assigned to the expansion.\footnote{Discreteness makes the definition general enough to handle the (generic) case where $\eth a$ is not smooth. At smooth points and using the approximations of Eq.~\eqref{eq:approx} or \eqref{eq:approxx}, discrete nonexpansion implies~\cite{Bousso:2024iry} $\Theta\leq 0$ or $\theta\leq 0$, respectively, for the usual quantum or classical expansions~\cite{Bousso:2015mna}. Similarly, Discrete Max-Focusing (Conj.~\ref{conj:qfc}) reduces to quantum focusing ($\Theta'\leq 0$~\cite{Bousso:2015mna}) or classical focusing ($\theta'\leq 0$~\cite{Wald:1984rg}). These more familiar notions of expansion and focusing can be substituted throughout: they are entirely sufficient for understanding all examples in this paper.} \emph{Past-nonexpanding} is defined by substituting $H^-$ for $H^+$ in Eq.~\eqref{eq:futnonexp}.
\end{defn}
\begin{conj}[Discrete Max-Focusing]\label{conj:qfc}
 Let $\eth a^+$ be the set of points where a wedge $a$ is future-nonexpanding. The future-lightsheet of $a$ is the null hypersurface
    \begin{equation}
        L^+(a)\equiv H^+(a') \cap J^+(\eth a^+)~.
    \end{equation}
    (The past lightsheet is defined analogously.)\\
    Let $a$, $b$, and $c$ be wedges such that $a\subset b\subset c$, and suppose that $\eth b \cup \eth c \subset \eth a \cup L^+(a)$. Then
    \begin{equation}\label{eq:qfc}
        \hmg(c|b)\leq 0~.
    \end{equation}
    The same statement holds in the past direction, i.e., if $\eth b \cup \eth c \subset \eth a \cup L^-(a)$.
\end{conj}

\begin{defn}[Antinormal]
A wedge $a$ is \emph{antinormal}\,
at $p\in \eth a$ if it is both future- and past nonexpanding at $p$.
\end{defn}


\begin{defn} \label{def:accessible}
     Given a wedge $a$, the wedge $f\supset a$ is said to be \emph{accessible from $a$} if it satisfies the following properties:
\begin{enumerate}[I.]
    \item $a\subset f\subset \tilde a '$;
    \item $f$ is antinormal at points $p\in \eth f\setminus\eth a$;
    \item $f\cap a'$ admits a Cauchy slice $\Sigma$ such that for any wedge $h\subsetneq f$ with $h\supset a$, $\eth h\subset \Sigma\cup \eth a$, and $\eth h\setminus \eth f$ compact in $M$,
    \begin{equation}
        \hmg(f|h)\leq 0~.
    \end{equation}
\end{enumerate}
Property III implies that the area of any intermediate edge $\eth h$ suffices as an entanglement resource for performing quantum state merging from $f$ to $h$~\cite{Akers:2020pmf}. 
\end{defn}

\subsection{Max- and Min-Holograms}
\label{sec:edef}

In this subsection, we modify the definition of generalized entanglement wedges, or holograms of bulk regions. Instead of holding fixed the ``conformal edge''~\cite{Bousso:2023sya}, we require that both the max- and min-hologram be accessible from $a$ using the new definition of accessibility given above. That is, both the max- and the min-hologram must be contained in $\tilde a'$, the complement of the conformal complement of $a$. 

\begin{defn}[Max-Hologram]\label{def:emax}
Given a wedge a, let $F(a)$ be the set of all wedges accessible from $a$.
The {\em generalized max-entanglement wedge} of $a$, or \emph{max-hologram} of $a$, $\emax(a)$, is their wedge union:
   \begin{equation}\label{eq:def:emax}
       \emax(a) \equiv \Cup_{f\in F(a)}\, f~.
   \end{equation}
\end{defn}

\begin{defn}[Min-Hologram]\label{def:emin} 
  Given a wedge $a$, let $G(a)\equiv \set{g: \mathrm{i}\, \wedge\, \mathrm{ii} \,\wedge\, \mathrm{iii}}$ be the set of all wedges that satisfy the following properties:\footnote{To make this look more like Def.~\ref{def:emax}, we could have defined a notion of min-accessibility that parallels Def.~\ref{def:accessible} of (max-)accessibility~\cite{Akers:2023fqr}. But in fact, the notions of min-accessibility can be largely eliminated~\cite{Bousso:2024iry} except here; and the complementarity theorem~\ref{thm:maxmincomp} will allow us to eliminate it altogether. We introduce the present definition to make contact with the prior literature, but it should ultimately be replaced by Eq.~\eqref{eq:comp}.}
  \begin{enumerate}[i.]
  \item  $a\subset  g\subset \tilde a'$;
  \item $g'$ is antinormal; 
  \item $g' \cap \tilde a'$ admits a Cauchy slice $\Sigma'$ such that for any wedge $h \neq g$ such that $g \subset h\subset \tilde a'$, $\eth h\subset \Sigma'$, and $\eth h\setminus \eth g$ is compact, 
  \begin{equation}
      \hmg(g'|h') \leq 0
  \end{equation}
  \end{enumerate}
  The {\em generalized min-entanglement wedge} or \emph{min-hologram} of $a$, $\emin(a)$, is their intersection:
  \begin{equation}\label{eq:emindef}
      \emin(a)\equiv \cap_{g\in G(a)}\, g~.
  \end{equation}
\end{defn}

\begin{figure}
    \centering
    \includegraphics[width=1\linewidth]{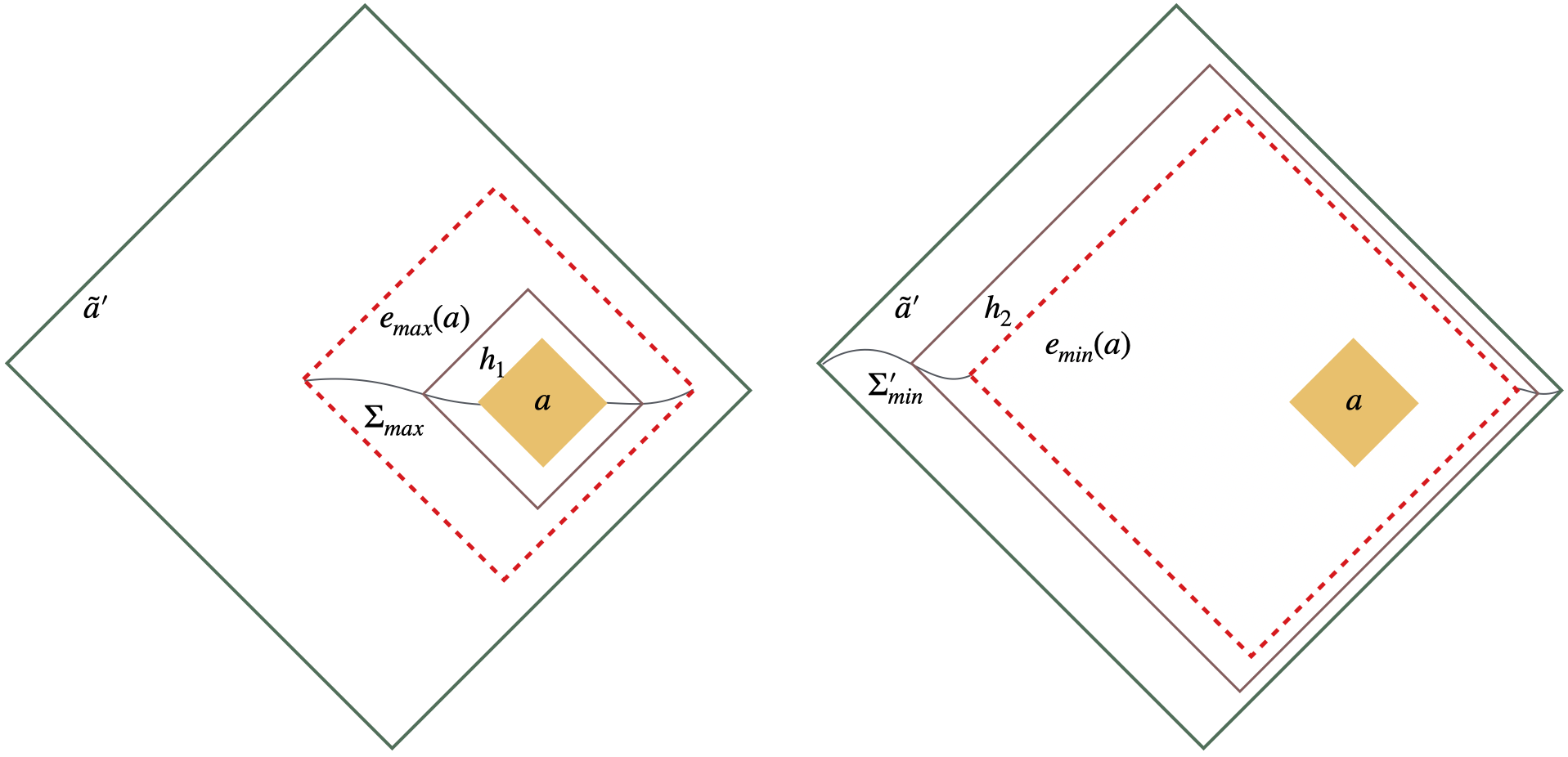}
    \caption{
    Depiction of property III of Def.~\ref{def:accessible} (left) and iii of  Def.~\ref{def:emin} (right). Min and max entanglement wedges are shown in red. On the right figure, the complement of the fundamental complement, $\tilde a'$, is outlined in green. }
    \label{fig:condition-iii-III}
\end{figure}

\section{Complementarity}\label{sec:prop}

\subsection{Properties of the Fundamental Complement}
\label{sec:fcprop}

\begin{lem}\label{lem:ta_anormal}
    For any wedge $a$, the fundamental complement $\tilde a$ is antinormal.
\end{lem}
\begin{proof}
    Let $\tau$ be a timelike curve that is future- and past-infinite in $a'$. Then the wedge $\tau''$ is the exterior of a past and a future causal horizon~\cite{Bousso:2025xyc}. By the Generalized Second Law~\cite{Bekenstein:1972tm,Wall:2010jtc,Bousso:2025xyc}, $\tau''$ is both FNE and PNE, i.e., $\tau''$ is antinormal. By its definition, $\tilde a$ is the union of such sets, so by Theorem 46 of Ref.~\cite{Bousso:2024iry}, $\tilde a$ is antinormal. (It also follows that $\tilde a$ itself is the exterior of a past and a future causal horizon.)
\end{proof}

\begin{lem} \label{tta_empty}
      Let $\tilde a$ be the fundamental complement of $a$ in $M$, and let $\left. \widetilde{(\tilde a)}\right|_{a'}$ be the fundamental complement of $\tilde a$ in $a'$. Then 
      \begin{equation}
          \left. \widetilde{(\tilde a)}\right|_{a'} = \varnothing. 
      \end{equation}
\end{lem}
\begin{proof}
By its definition, $\tilde a$ contains all timelike curves that are both past- and future-infinite in $a'$. Hence its complement in $a'$, $(\tilde a)'|_{a'}$, contains no such curves.
\end{proof}

\subsection{Complementarity of Holograms}
\label{sec:comp}

\begin{defn}[Restricted Max-Hologram] \label{def:restrictedemax}
    Let $a\subset b$ be wedges. We define $\emax(a)|_b$ as $\emax(a)$ computed in the spacetime $b$, except that expansions are evaluated within the spacetime $M$. More explicitly, let $F(a|b)$ denote the set of all wedges satisfying the following properties:
    \begin{enumerate}[I.]
    \item $a\subset f\subset (\tilde a|_b)' \cap b$;
    \item $f$ is antinormal at points $p\in \eth f\setminus\eth a$ in $M$;
    \item $f\cap a'$ admits a Cauchy slice $\Sigma$ such that for any wedge $h\subsetneq f$ with $h\supset a$, $\eth h\subset \Sigma\cup \eth a$, and $\eth h\setminus \eth f$ compact in $M$,
    \begin{equation}
        \hmg(f|h)\leq 0~.
    \end{equation}
\end{enumerate}
where in property (I) one uses null infinity $\scri|_b = \scri\cap [\cl b]_{\bar M}$ to compute the fundamental complement of $a$ in $b$, $\tilde a|_b$.
\end{defn}
The  \emph{restricted max-hologram} of $a$, $\emax(a)|_b$, is the union:
  \begin{equation}\label{eq:emindef}
      \emax(a)|_b\equiv \Cup_{f\in F(a|b)}\, f~.
  \end{equation}

\begin{thm}[Complementarity of holograms]\label{thm:maxmincomp} 
Let $a$ be a wedge in $M$, let $a'$ be its complement wedge in $M$, and let $\tilde a$ be the fundamental complement of $a$ in $M$; see Definitions~\ref{def:sc} and \ref{def:fc}. Then the min-hologram $\emin(a)$ is the complement, in $M$, of the restricted max-hologram $\emax(\tilde a)|_{a'}$ (see Def.~\ref{def:restrictedemax}):
\begin{equation}\label{eq:comp}
    \emin(a)' = \emax(\tilde a)|_{a'} 
    \end{equation}
\end{thm}
\begin{proof}
Note that we have 
    \begin{equation}
    \begin{split}
         \emin(a)' = \Cup_{g\in G(a)}\, g' \: , \quad (\emax(\tilde a)|_{a'})' = \cap_{f\in F(\tilde a |a')}\, f' \: .
    \end{split}
\end{equation}

First we will show that $(\emax(\tilde a)|_{a'})' \supset \emin(a)$. Since $\emin(a)=\cap_{g\in G(a)} g$ and $(\emax(\tilde a)|_{a'})' = \cap_{f\in F(\tilde a |a')}\, f'$, it suffices to show that for every $f\in F(\tilde a |a')$, $f'$ is an element of $G(a)$, i.e., that $f'$ satisfies the properties i--iii that characterize $G(a)$ in Def.~\ref{def:emin}. 

Property i: By property I of $f$, $f \supset \tilde a $, and so $f' \subset \tilde a'$. Since $F(\tilde a |a')$ is constructed in the spacetime $a'$, we have $f\subset a'$ which implies $f'\supset a$ in the spacetime $M$ in which $G(a)$ is constructed. Property ii: By Lemma \ref{lem:ta_anormal}, $\tilde a$ is antinormal, so $f=f''$ is antinormal by Property II and strong subadditivity. 
Property iii: The Cauchy slice $\Sigma$ of $f\cap a'$ guaranteed by property III of Def.~\ref{def:accessible} (with the substitution $a\to\tilde a$ in that definition) satisfies property $iii$ of Def.~\ref{def:emin} (with the substitution $g\to f'$ in that definition).

Next we will show that $\emin(a)'\subset \emax(\tilde a)|_{a'}$. Since $\emin(a)' = \cup_{g\in G(a)}\, g'$ and $\emax(\tilde a)|_{a'} = \cup_{f\in F(\tilde a|a')} f$, it suffices to show that
for every $g\in G(a)$,  $g'\in F(\tilde a|a')$, i.e., that $g'$ satisfies the properties I--III that characterize $F(a|a')$ in Def.~\ref{def:accessible} and \ref{def:restrictedemax}. 

Property I: By property i of $g$, $a'\supset g'\supset \tilde a$. By Lemma \ref{tta_empty}, $\left. \widetilde{(\tilde a)}\right|_{a'}$ is empty, so $g'\subset \left. \widetilde{(\tilde a)}\right|_{a'}'$ trivially. Property II: $g'$ is antinormal by property ii. Property III: the Cauchy slice of $g'\cap\tilde a'$ guaranteed by property iii of Def.~\ref{def:emin} satisfies property $\mathrm{III}$ of Def.~\ref{def:emax} (after substituting $f\to g'$ and $a\to \tilde a$ in the latter definition).
\end{proof}

\section{Examples of Fundamental Complements and Holograms} \label{sec:examples}

In this section, we will display many examples of fundamental complements, entanglement wedges, and complementarity in various classes of spacetimes. We will begin by discussing Big Bang (or Big Crunch) cosmologies; then we will turn to asymptotically flat and asymptotically de Sitter spacetimes; and finally, we will recover entanglement wedges of AdS boundary regions as a special case. 

We will only consider examples where $\emin(a)=\emax(a)$ and $\emin(\tilde a) = \emax(\tilde a)$, and we will therefore abbreviate $\emin$ and $\emax$ as $e$. For an example with $\emin(a)\supsetneq\emax(a)=a$, see Fig.~\ref{fig:AdS_2sided}.

The notion of null infinity, $\scri$, will be useful to our discussion, so we begin by reviewing it here. We emphasize, however, that our definition of the fundamental complement $\tilde a$, and thus of $\emin(a)$ and $\emax(a)$, is entirely intrinsic to the spacetime $M$. It does not require the explicit construction of a conformal boundary.

\begin{defn}
    The unphysical spacetime $(\bar M,\bar g)$  is the conformal completion~\cite{Wald:1984rg} of $(M,g)$ by addition of a \emph{conformal boundary}, $\partial_{\bar M} M\equiv \bar M\setminus M$.
\end{defn}

\begin{defn}
    \emph{Future infinity}, $\scri^+$, is the subset of $\partial M$ consisting of the future endpoints in $\bar M$ of null geodesics of future-infinite affine length in $M$. \emph{Past infinity}, $\scri^-$, is defined similarly. \emph{Null infinity} is their union:
    \begin{equation}
        \scri \equiv \scri^+ \cup \scri^-~.
    \end{equation}
\end{defn}

\begin{rem}
    When $\scri^\pm$ are defined as above, and thus in all examples in this section, the fundamental complement of $a$ is given by 
    \begin{equation}
        \tilde a = \left([I^+(\scri^-|_ {a'})]\cap [I^-(\scri^+|_ {a'})]\right)''~,
    \end{equation}
    where the subscripts remind us to include only the null infinities of $a'$.
\end{rem}

\subsection{Big Bang and Big Crunch Cosmologies}
\label{sec:cosmo}

Here we show that any universe with a Big Bang (or with a Big Crunch) is trivially reconstructible, regardless of its spatial topology. Trivial reconstructibility (``a one-dimensional Hilber space'') is a property traditionally associated with spatially closed universes in the literature~\cite{Marolf:2020xie, Usatyuk:2024mzs, McNamara:2020uza, Abdalla:2025gzn, 
 Almheiri:2019hni, Harlow:2025pvj}. But in fact, it suffices for $M$ to contain no timelike curves that have infinite duration both to the past and to the future. Perhaps just as surprisingly, we will find in Sec.~\ref{sec:dS} that de Sitter space, despite being spatially closed, is \emph{not} trivially reconstructible.

\begin{thm}\label{thm:ta_empty}
    If $\tilde a= \varnothing$, then $\emax(a)=\emin(a)=M$.
\end{thm}
\begin{proof}
    Note that $M\in F(a)$: accessibility condition I is satisfied since $M\supset a $; II is satisfied because $M$ is trivially antinormal; and condition III is satisfied because for any wedge $h$, $\hmg(M|h)\leq 0$. Since $\emax$ is the union of all wedges in $F$, it follows that $\emax(a)=M$. By Theorem~\ref{thm:emaxemin}, $\emin(a)=M$. [This can also be shown directly: $M\in G(a)$ because $\tilde a =\varnothing$, $M'=\varnothing$ is trivially antinormal and satisfied condition iii of Def.~\ref{def:emin}. But condition iii fails for any other wedge $g\supset a$, since $\hmg(g'|M')>0$. Hence $G(a)=\{M\}$ and $\emin(a)=M$.]
\end{proof}

\begin{cor}\label{cor:trivial}
    Suppose that the spacetime $M$ contains no future-infinite timelike curve, or that $M$ contains no past-infinite timelike curve. (Loosely speaking this means that $M$ starts with a Big Bang, or ends with a Big Crunch.) Let $a\subset M$ be a wedge. Then $\emax(a)=\emin(a)=M$. 
\end{cor}
\begin{proof}
    $M$ contains no timelike curve that is both future-infinite and past-infinite, so neither does $a'$. Therefore $\tilde a=\varnothing$ and the result follows from the preceding theorem.
\end{proof}

\subsection{Asymptotically Flat Spacetimes}
\label{sec:flat}

In asymptotically flat spacetimes, the restriction to $\tilde a'$ gives rise to interesting constraints. Examples are shown in Figures~\ref{fig:AFlat-annulus}-\ref{fig:AFlat-rindler-inters}. In particular, we obtain a physically sensible entanglement wedge for Rindler wedges and local deformations of Rindler wedges:
\begin{defn}\label{def:rindler}
Let $a\neq \varnothing$ be a wedge in an asymptotically flat spacetime $M$. We call $a$ a \emph{Rindler wedge} if $a=\{p,q\}''$, where $p\in\scri^-$ and $q\in\scri^+$.
\end{defn}

\begin{figure}[htbp]
    \centering
    \includegraphics[width=1\linewidth]{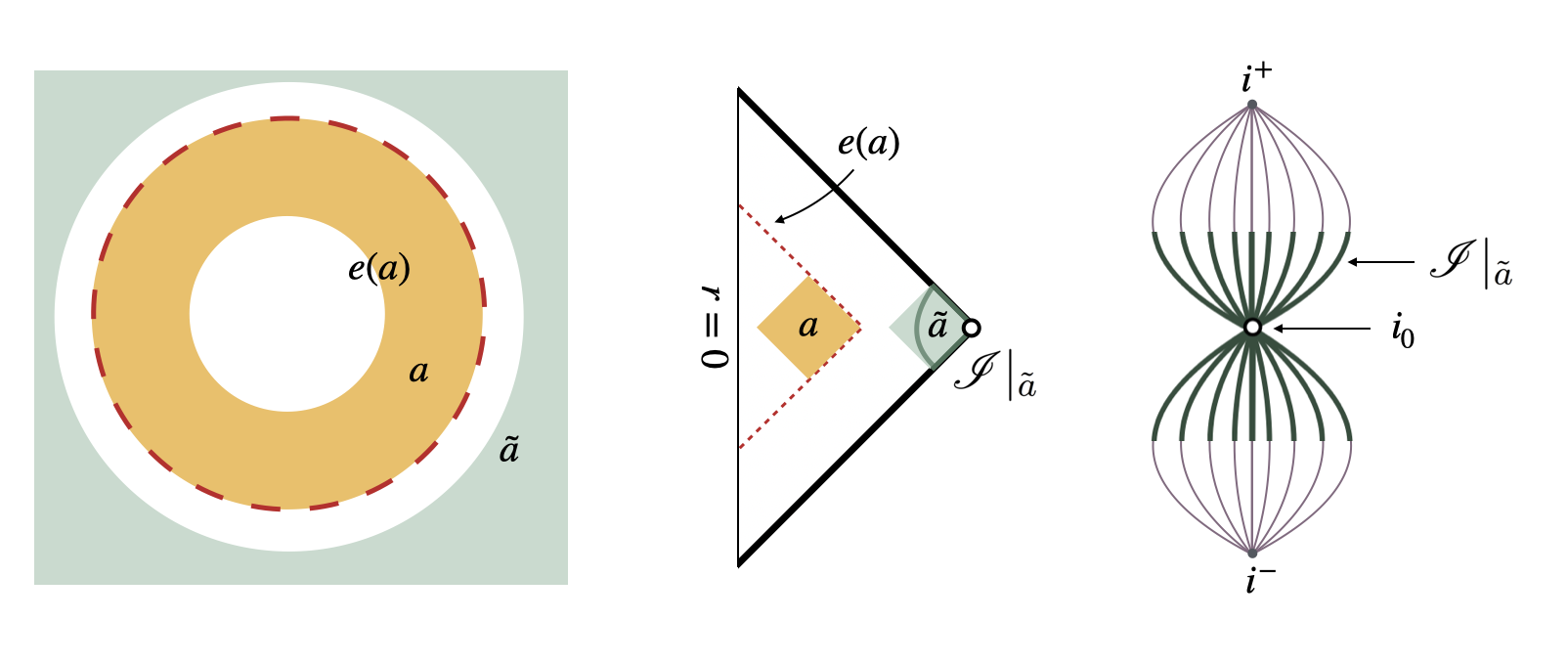}
    \caption{ The entanglement wedge $e(a)$ of a shell $a$ in asymptotically flat spacetime includes its interior. The only role of $\tilde a$ to prevent $e(a)=M$. \emph{Left:} Cauchy slice containing the edges of $a$, $e(a)$, and $\tilde a$. Middle: Penrose diagram of the configuration. \emph{Right:} The view of $\scri$. Each curve connecting $i_0$ to $i^\pm$ represents  a null generator on $\scri$. In the presence of generic matter, focusing leads to $\tilde a' \supsetneq e(a)$.}
    \label{fig:AFlat-annulus}
\end{figure}

\begin{figure}[H]
    \centering
    \includegraphics[width=1\linewidth]{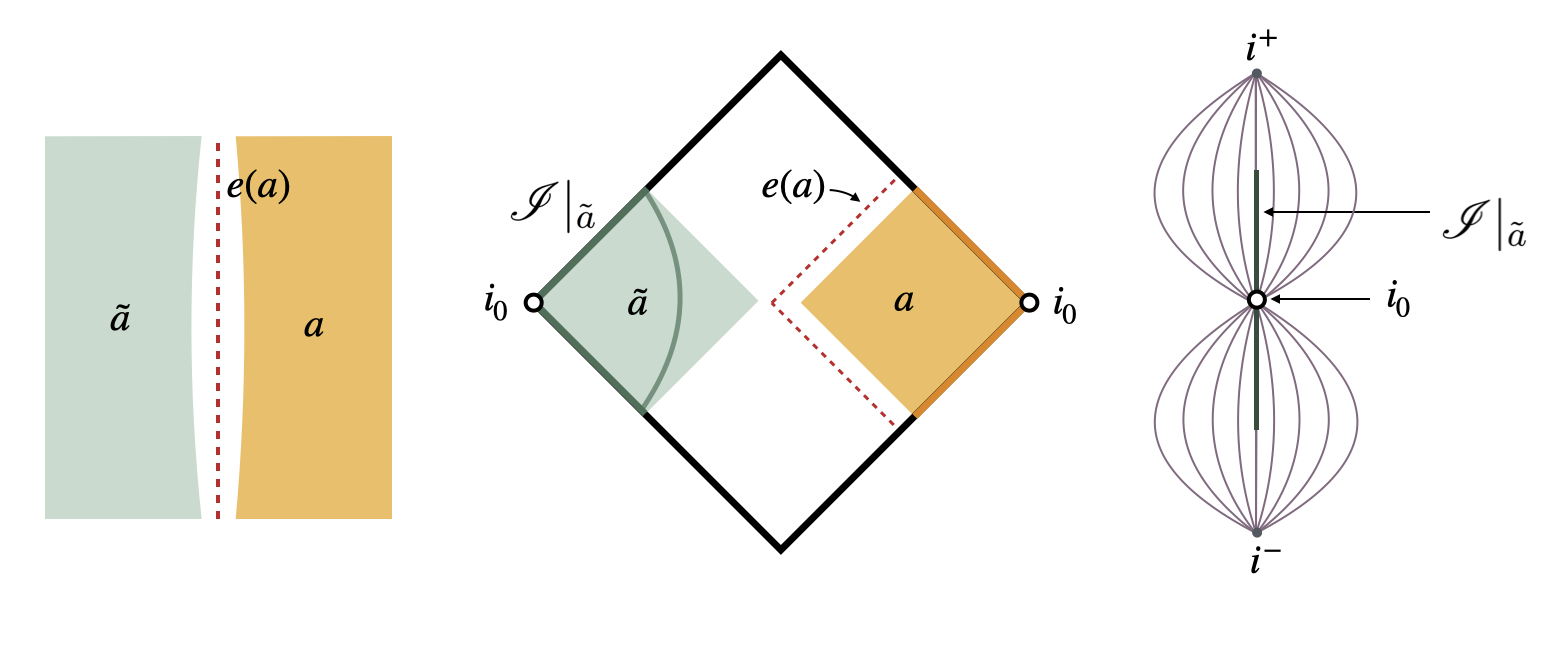}
    \caption{Here $a$ is a Rindler wedge in asymptotically flat spacetime with generic matter content. The matter focuses the Rindler horizons, making $a$ properly antinormal (concave), and causing the gaps between the edges of $\tilde a$, $e(a)$, and $a$. Only $e(a)$ is bounded by an extremal surface; $\tilde a$ and $a$ are both antinormal. In exact vacuum Minkowski space, the gap would close, so that $e(a) = a =\tilde a'$. \emph{Left:} Portion of a Cauchy slice containing all edges. \emph{Middle:} Penrose diagram. \emph{Right:} Conformal infinity. Each curve connecting spacelike infinity $i_0$ to the timelike infinities $i^\pm$ is a generator of null infinity $\scri$. }
    \label{fig:AFlat-Rindler}
\end{figure}

\begin{figure}[H]
    \centering
    \includegraphics[width=1\linewidth]{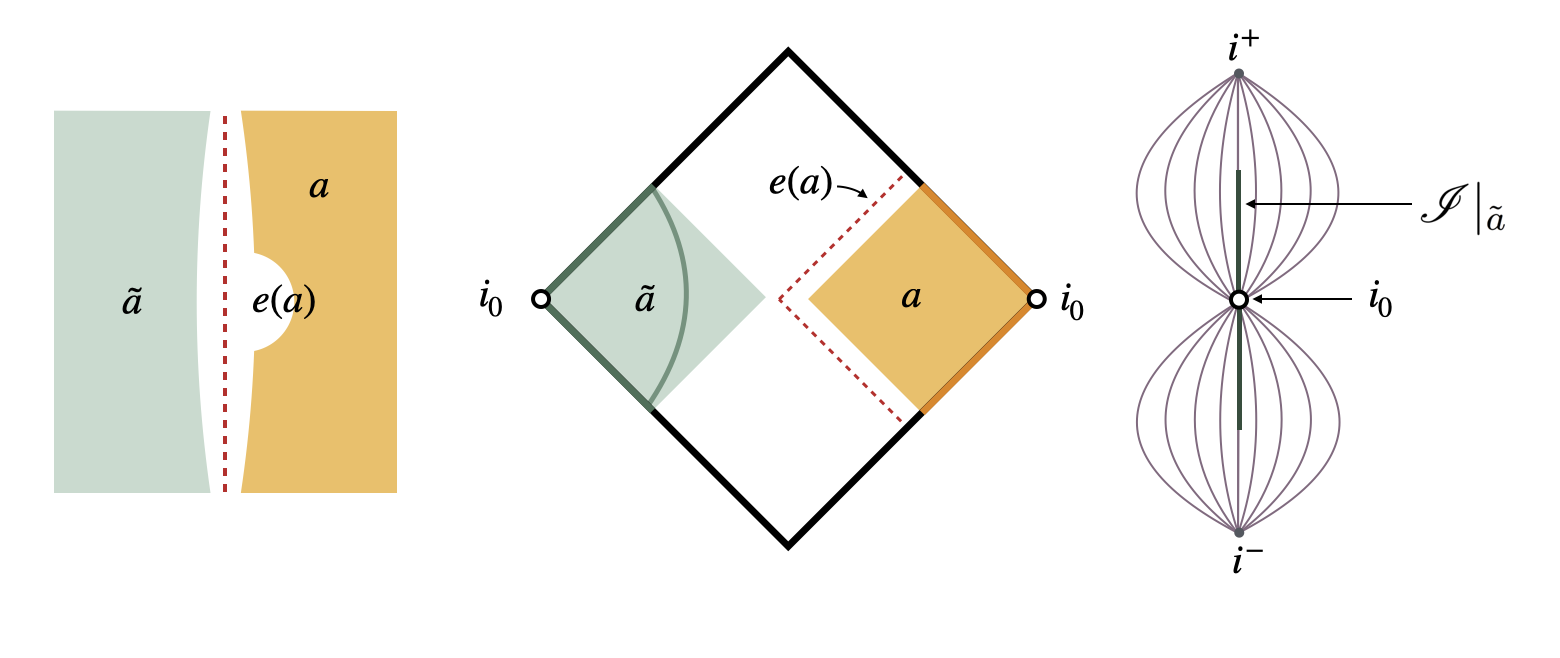}
    \caption{Here $a$ is a Rindler wedge with a notch (a spacelike inward deformation. Again the gaps are generic, but $e(a)$ always fills in the notch. In exact Minkowski, $\tilde a'=e(a)$ would be an undeformed Rindler wedge. (The Penrose diagram does not show the notch.)
    }
    \label{fig:AFlat-Rindler-notch}
\end{figure}
 
\begin{figure}[H]
    \centering
    \includegraphics[width=1\linewidth]{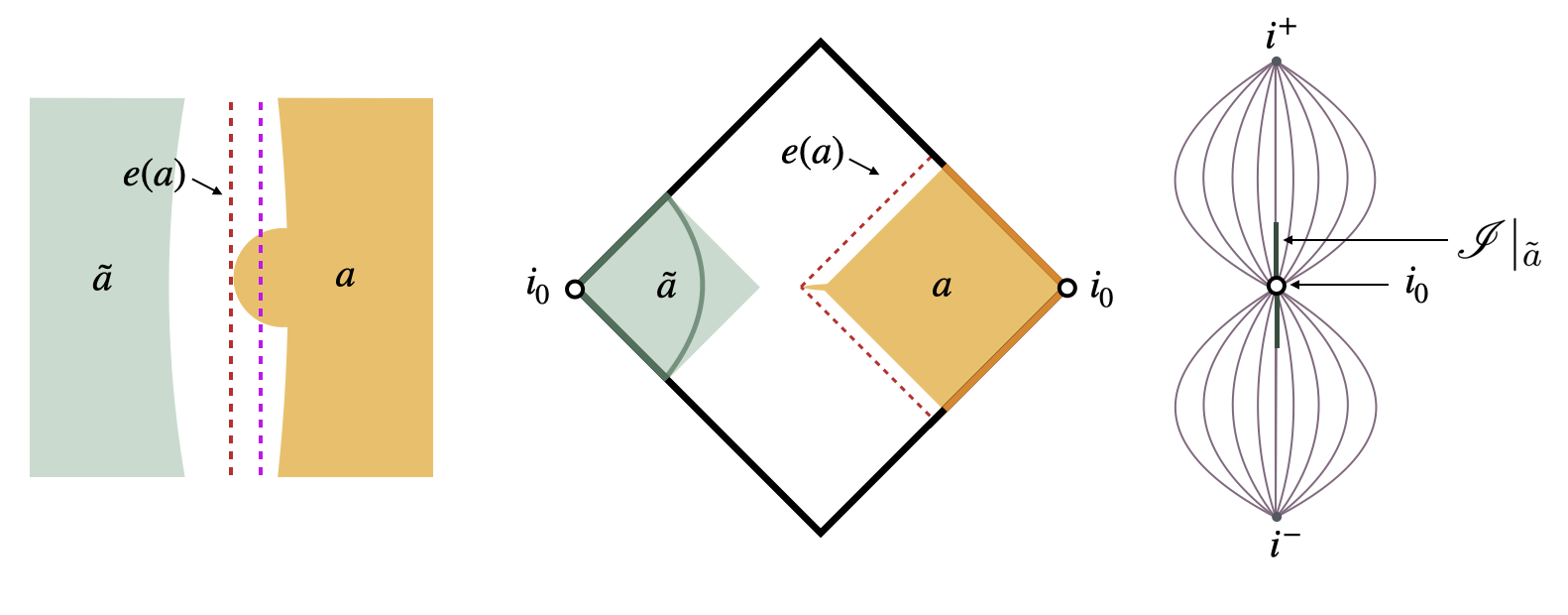}
    \caption{
    Here $a$ is a Rindler wedge with a bump (an outward deformation). Here $\eth e(a)$ is an extremal surface that touches the bump. The pink dashed line indicates the entanglement wedge that $a$ would have had if the bump was removed. The Penrose diagram is somewhat schematic due to lack of transverse symmetry; the bump is symbolized by the line protruding from $a$. On the right, note that the conformal boundary of $\tilde a$ is shortened compared to the previous two examples.}
    \label{fig:AFlat-Rindler-bump}
\end{figure}

\begin{figure}[H]
    \centering
    \includegraphics[width=0.7\linewidth]{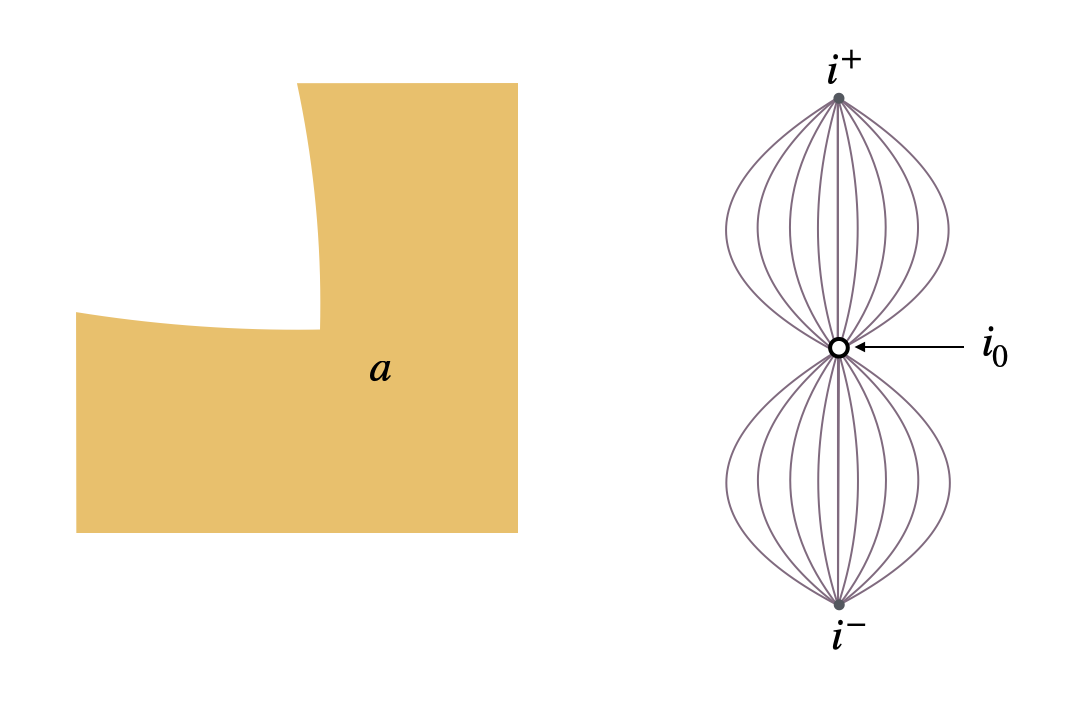}
    \caption{Here $a$ is the union of two Rindler wedges. All infinite timelike curves enter the past and future $I(a)$, so $\tilde a=\varnothing$ and $e(a)=M$.}
    \label{fig:AFlat-rindler-union}
\end{figure}

\begin{figure}[H]
    \centering
    \includegraphics[width=0.7\linewidth]{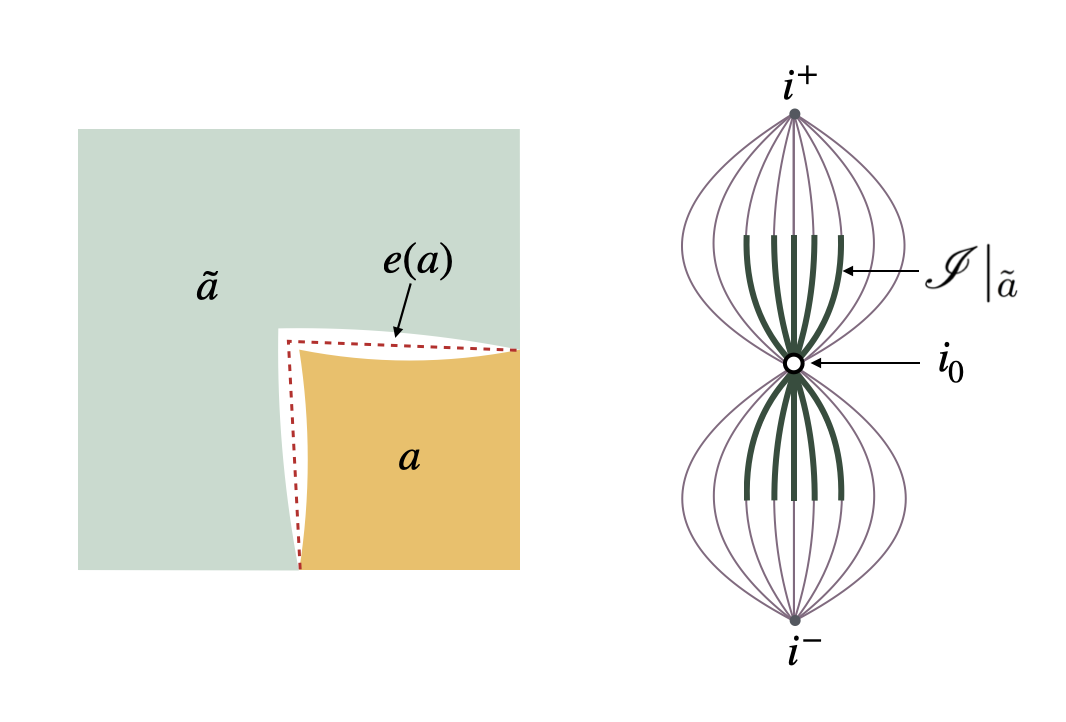}
    \caption{Here $a$ is the intersection of two Rindler wedges. In this case $\tilde a$ is the wedge union of a whole ``arc'' of Rindler wedges. More precisely, $\scri|_{\tilde a}$ is generated by an arc (or solid angle) of half-generators of $\scri^\pm$. 
    }
    \label{fig:AFlat-rindler-inters}
\end{figure}

\subsection{Asymptotically  de Sitter Spacetimes}
\label{sec:dS}

We will now exhibit a few examples in asymptotically de Sitter spacetimes. For a closed recollapsing universe $M$ in a pure state, $\emax(a)=M$ for any nonempty $a$. However in de Sitter, we will see that the presence of $\scri$ can prevent this trivial outcome. See Figs.~\ref{fig:dS-examples} and \ref{fig:dS-sch-ds} for some interesting examples. 
\begin{figure}[ht!]
    \centering
    \includegraphics[width=1\linewidth]{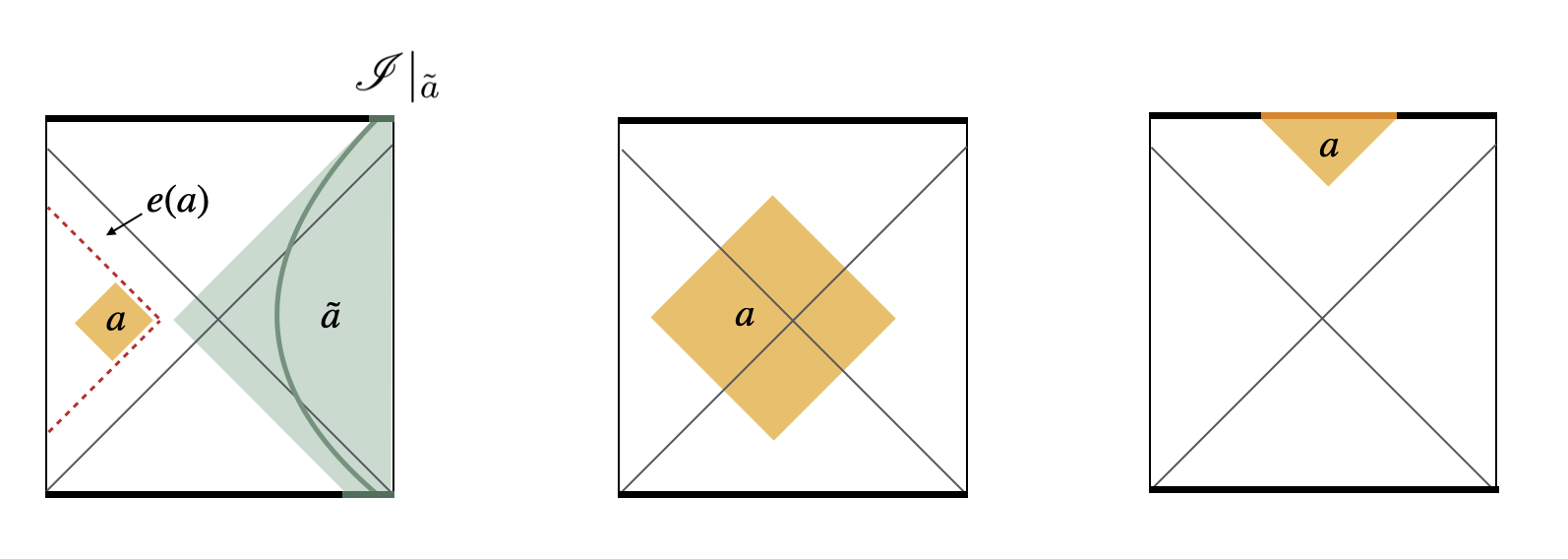}
    \caption{\emph{Left:} $a$ is a shell contained in a single static patch of de Sitter space. The fundamental complement $\tilde a$ prevents $e(a)$ from including the entire universe.  \emph{Middle and Right:} If the past and future of $a$ includes past or future infinity, then $\tilde a=\varnothing$ and $e(a)=M$. 
    }
    \label{fig:dS-examples}
\end{figure}

\begin{figure}[ht!]
    \centering
    \includegraphics[width=1\linewidth]{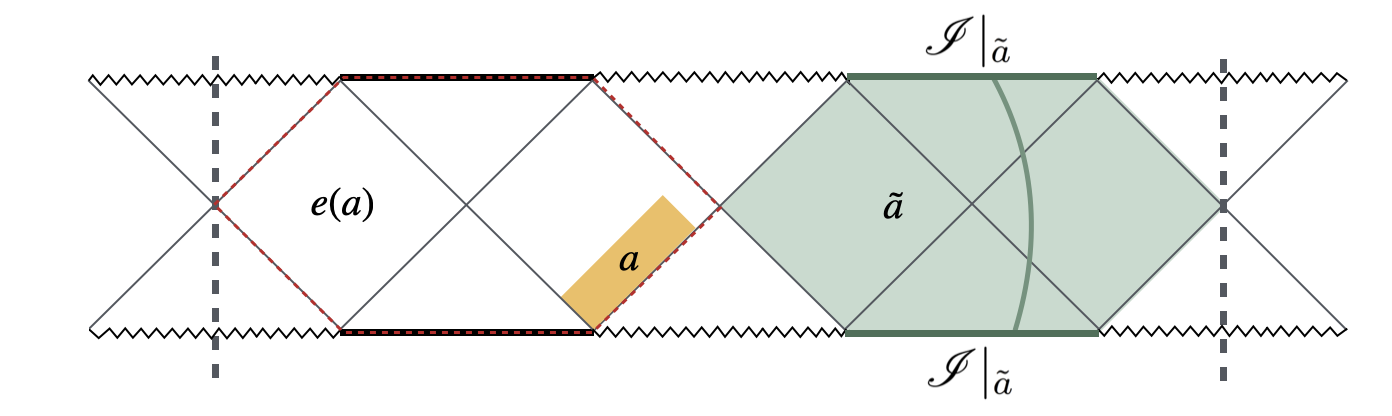}
    \caption{Schwarzschild-de Sitter spacetime; the wedge $a$ lies between the black hole and cosmological horizons. The figure shows the exact vacuum case where $e(a) = \tilde a'$. With generic matter added, the apparent horizons would become spacelike and a gap would open up. 
    }
    \label{fig:dS-sch-ds}
\end{figure}

\section{AdS Entanglement Wedges as a Special Case}
\label{sec:ads}

In this section, we will describe how the entanglement wedges of an AdS boundary region $B$ can be recovered from the prescription given in Sec.~\ref{sec:edef}, which takes a bulk region $a$ as input. The answer is to apply our prescription to the bulk wedge $\CW(B)$, the causal wedge of $B$.

The max- and min-entanglement wedges, $\maxEW(B)$ and $\minEW(B)$, of an AdS boundary wedge $B\subset \scri$ are defined in Ref.~\cite{Akers:2023fqr}, building on extensive earlier work~\cite{Ryu:2006bv,Hubeny:2007xt,Faulkner:2013ana,Engelhardt:2014gca,Akers:2020pmf}. For full generality, we will not assume that all of $\scri$ is present but only some globally hyperbolic subset of $\scri$. For example, only a small timeband may be present.

\begin{defn}\label{def:ew}
Let $M$ be asymptotically AdS, and let $B\subset \scri$ be a wedge on conformal infinity. Let $\bar B$ be the spacelike complement of $B$ on $\scri$. Let $F(B)\equiv \set{f:\mathrm{I}\,\wedge\,\mathrm{II}\,\wedge\,\mathrm{III}}$ be the set of all wedges in $M$ that satisfy the following properties:
\begin{enumerate}[1.]
\item $B$ is the conformal infinity of $f$;
\item $f$ is antinormal;
\item $f$ admits a Cauchy slice $\Sigma_{\maxEW}$ such that for any wedge $h \neq f$ with conformal boundary $B$ and $\eth h\subset \Sigma$,
    \begin{equation}
        \hmg(h|f)\leq 0~.
    \end{equation}
\end{enumerate}
The max-entanglement wedge $\maxEW(B)$ is their wedge union:
   \begin{equation}
       \maxEW(B) \equiv \Cup_{f\in F(B)}\, f~.
   \end{equation}
Let $G(B)\equiv \set{g: \mathrm{i}\, \wedge\, \mathrm{ii} \,\wedge\, \mathrm{iii}}$ be the set of all wedges that satisfy the following properties: 
    \begin{enumerate}[i.]
        \item $B$ is the conformal infinity of $g$\,;
        \item $g'$ is antinormal;
        \item $g'$ admits a Cauchy slice $\Sigma'_{\minEW}$ such that for any wedge $h \neq g$ with conformal boundary $B$ and $\eth h\subset \Sigma'$, 
        \begin{equation}
            \hmg(g'|h')\leq 0~.
        \end{equation}.
    \end{enumerate} 
  The min-entanglement wedge $\minEW(B)$ is their intersection:
   \begin{equation}
       \minEW(B) \equiv \cap_{g\in G(B)}\, g~.
   \end{equation} 
\end{defn}


\begin{defn}[Causal wedge of a boundary region]
    Given a wedge $B\subset \scri$ in asymptotically AdS,  
    the \emph{causal wedge} of $B$ is the double spacelike complement of $B$ in the bulk:
    \begin{equation}
        \CW(B)\equiv (B'_{\bar M}\cap M)'
    \end{equation}
    To be precise, one takes the spacelike complement of $B$ in the conformal completion $\bar M$ of $M$; then one retains only its intersection with $M$ (and thus one discards $\bar B$); and finally one takes another spacelike complement in $M$; this yields $\CW(B)$. (The fact that this yields a wedge follows from Corollary~\ref{cor:comp_open} in Appendix~\ref{app:largecomp}.)
\end{defn}

\begin{thm}
    Let $M$ be asymptotically AdS with null infinity $\scri$. ($\scri$ can be a globally hyperbolic subset of the full conformal boundary of AdS.) Let $B\subset \scri$ be a boundary wedge. Then
    \begin{equation}
        \maxEW(B) = \emax[\CW(B)]~,~~~\minEW(B) = \emin[\CW(B)]~.
    \end{equation}
\end{thm}

\begin{proof}
    We will show that $\maxEW(B)$ satisfies properties I-III of Def.~\ref{def:emax}, which implies $\maxEW(B)\subset \emax[\CW(B)]$. Conversely, we will show that $\emax[\CW(B)]$ satisfies properties I-III of Def.~\ref{def:ew}, which implies $\emax[\CW(B)]\subset \maxEW(B)$. Similarly, we will prove mutual inclusion for $\emin[\CW(B)]$ and $\minEW(B)$.

    \emph{Properties I and i}: The fundamental complement of $\CW(B)$ in the bulk is easily seen to satisfy 
    \begin{equation}\label{eq:cwb}
        \widetilde{\CW(B)}  = \CW(\bar B)
    \end{equation}
    by Def.~\ref{def:fc}.
    This ensures that $\emax[\CW(B)]$ and $\emin[\CW(B)]$ both have conformal boundary $B$ and thus satisfy the homology conditions I and i required of $\maxEW(B)$ and $\minEW(B)$ in Def.~\ref{def:ew}. 
    
    Conversely, causal wedge inclusion~\cite{Wall:2012uf,Akers:2023fqr} implies that $\maxEW(B)$ and $\minEW(B)$ both contain $\CW(B)$. When combined with traditional entanglement wedge complementarity, Eq.~\ref{eq:oldcomp}, causal wedge inclusion also implies that $\maxEW(B)$ and $\minEW(B)$ are both contained in $\CW(\bar B)'$. Hence $\maxEW(B)$ satisfies property I of Def.~\ref{def:emax} of $\emax[\CW(B)]$, and $\minEW(B)$ satisfies property i of Def.~\ref{def:emin} of $\emin[\CW(B)]$. 
    
    \emph{Properties II and ii}: $H^+(B)$ is a future causal horizon and $H^-(B)$ is a past causal horizon; thus $\CW(B)$ is antinormal by the Generalized Second Law~\cite{Bekenstein:1972tm,Wall:2010jtc,Bousso:2025xyc}. By Lemma~\ref{lem:ssa} and Theorem~\ref{thm:emaxaccessible},  $\emax[\CW(B)]$ is antinormal.
    
    Conversely, $\maxEW(B)$ is antinormal by a straightforward analogue~\cite{Akers:2023fqr} of Theorem~\ref{thm:emaxaccessible}. 

    Using Eq.~\eqref{eq:cwb}, the analogous results for $\emin[\CW(B)]$ and $\minEW[\CW(B)]$ follow from the complementarity relations~\eqref{eq:oldcomp} and \eqref{eq:comp}.

    \emph{Properties III and iii:} The future Cauchy horizon of a causal wedge is a causal horizon and satisfies the generalized second law~\cite{Bousso:2025xyc}. By strong subadditivity, Eq.~\eqref{eq:ssa}, the Cauchy surfaces $\Sigma$ and $\Sigma'$ appearing in Def.~\ref{def:emax} of $\emax[\CW(B)]$ and in Def.~\ref{def:emin} of $\emin[\CW(B)]$, respectively, can be completed to the Cauchy slices appearing in properties III and iii of Def.~\ref{def:ew} by adjoining $H^+[CW(B)]$ to $\Sigma$ and $H^+[\CW(\bar B)]$ to $\Sigma'$.

    Conversely, given the Cauchy slices $\Sigma_{\maxEW}$ and $\Sigma_{\minEW}$ guaranteed to exist by properties III and iii of Def.~\ref{def:ew}, let $\Sigma = \Sigma_{\maxEW}\cap \CW(B)'\cup J^-[\Sigma_{\maxEW}] \cap H^+[\CW(B)'] \cup J^-[\Sigma_{\maxEW}] \cap H^+[\CW(B)']$, and let $\Sigma' = \Sigma_{\minEW}\cap \CW(\bar B)'\cup J^-[\Sigma_{\minEW}] \cap H^+[\CW(\bar B)'] \cup J^-[\Sigma_{\minEW}] \cap H^+[\CW(\bar B)']$. By strong subadditivity and the generalized second law, $\Sigma$ satisfies property III of Def.~\ref{def:emax}, and $\Sigma'$ satisfies property iii of Def.~\ref{def:emin}.
\end{proof}

Interesting examples are shown Figures \ref{fig:AdS-reflecting} and \ref{fig:AdS-non-reflectting}.

\begin{figure}[ht!]
    \centering
    \includegraphics[width=1\linewidth]{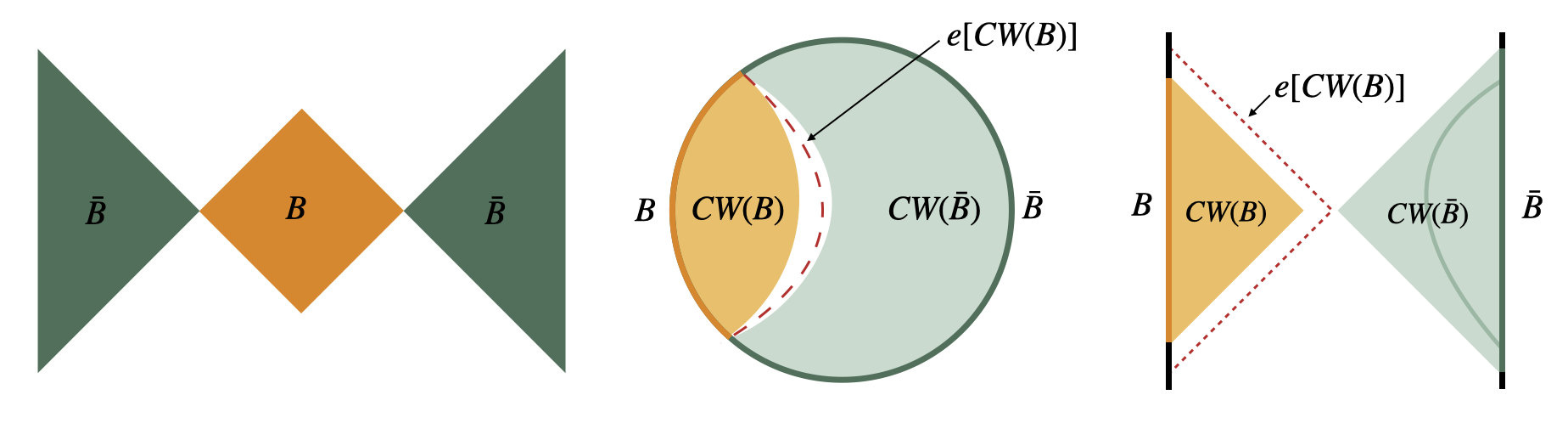}
    \caption{
    The entanglement wedge of an AdS boundary region $B$ can be recovered from our more general definition: $\EW(B) = e[\CW(B)]$, where $\CW$ is the causal wedge of $B$ in the bulk. The fundamental complement of $\CW(B)$ is the causal wedge of $\bar B$: $\widetilde{\CW(B)}=\CW(\bar B)$. As usual, we show a generic case with matter present, so that the entanglement wedge is larger than $\CW(B)$, and $\CW(\bar B)'$ is yet larger.}
    \label{fig:AdS-reflecting}
\end{figure}

\begin{figure}[H]
    \centering
    \includegraphics[width=1\linewidth]{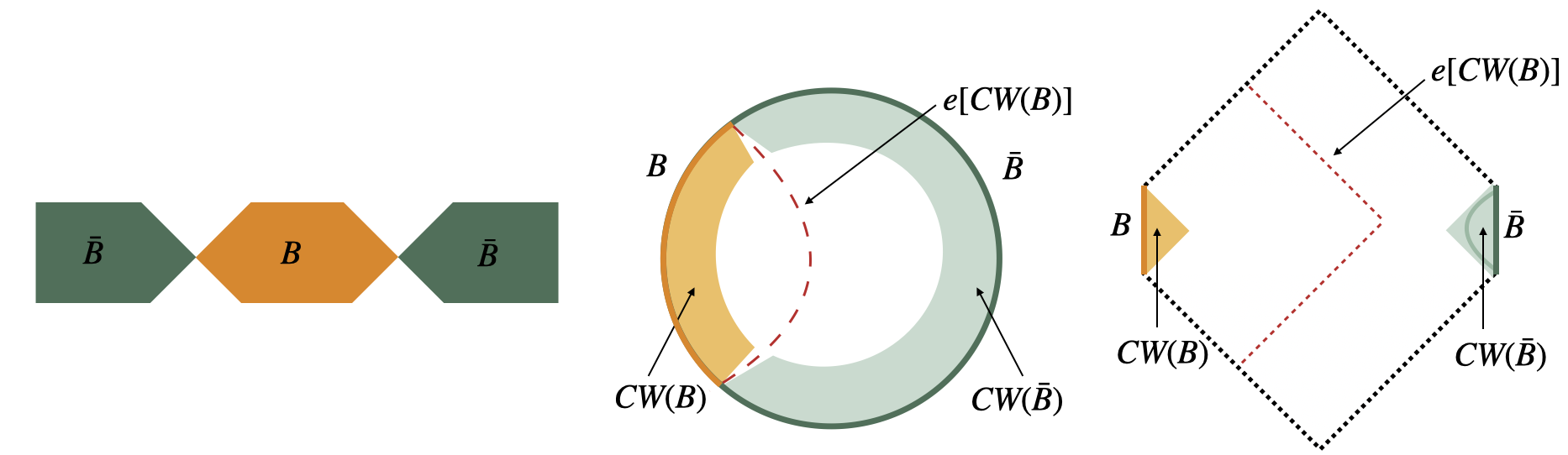}
    \caption{
    Only a timeband of $\scri$ is present in this example (left), so the bulk dual is an extendible manifold bounded by the dotted lines (right). Again, the ``traditional'' AdS/CFT entanglement wedge $\maxEW(B)$ can be recovered from our prescription as $\emax[\CW(B)]$. 
    }
    \label{fig:AdS-non-reflectting}
\end{figure}

\subsection*{Acknowledgments}
This work was supported in part by the Berkeley Center for Theoretical Physics; and by the Department of Energy, Office of Science, Office of High Energy Physics under Award DE-SC0025293.
 
\appendix

\section{Nesting, No-Cloning, and Strong Subadditivity of Holograms}
\label{app:oldthms}

In this appendix, we use the new definitions~\ref{def:accessible}, \ref{def:emax}, and \ref{def:emin} that involve the fundamental complement $\tilde a$, in order to re-prove important properties of bulk entanglement wedges that were established in Refs.~\cite{Bousso:2023sya,Bousso:2024iry} for the original definitions.\footnote{Throughout this Appendix, the large fundamental complement $\hat a$ (see Appendix~\ref{app:largecomp}) could be substituted for $\tilde a$; all proofs remain valid.} If a proof is unmodified compared to Ref.~\cite{Bousso:2023sya} or~\cite{Bousso:2024iry}, we provide only the reference. If a proof must be modified, then we produce the entire proof even if the modifications are minor. 

We begin by providing additional definitions that are needed only in this Appendix.

\begin{defn}\label{def:hmin_gen}
The \emph{generalized smooth min-entropy} of $b$ conditioned on $a$ is defined by
\begin{equation}
      \hmingen(b|a) = - \hmg(a'|b') ~.
\end{equation}
\end{defn}

\begin{defn}[Noncontracting Wedges]
    A wedge $a$ is said to be \emph{future-noncontracting} (FNC) if there exists an open set $O\supset\eth a$ such that no proper past-directed outward null deformation of $a$ with compact support within $O$ is FNE on all new edge points. That is, no wedge $b\supsetneq a$ with $\eth b\subset \eth a \cup [H^-(a')\cap O]$ is FNE on all points in $\eth b\setminus\eth a$. (Heuristically, in sufficiently smooth settings, FNC implies nonnegativity of the outward future quantum expansion of $a$; moreover, positivity of the quantum expansion implies FNC.) \emph{Past-noncontracting} (PNC) is defined analogously.   
\end{defn}

\begin{defn}[Marginal accessibility from $a$]
Let $a\subset k$ be wedges. $k$ is called \emph{future-marginally accessible} from $a$ if $k$ is accessible from $a$ and $k$ is FNC. \emph{Past-marginally accessible} is defined analogously. The wedge $k$ is called a \emph{throat accessible from $a$} if $k$ is future- and past-marginally accessible from $a$.     
\end{defn}

We now turn to deriving key properties of $\emin$ and $\emax$.

\begin{thm}
    [Chain rules]\label{conj:chain}
    Let $a\supset b\supset c$. The generalized max and min entropies obey the following inequalities:
    \begin{align}
        \hmg(a|c) &\leq\hmg(a|b)+\hmg(b|c) \label{eq:chainMaxMaxMax}\\
        \hmingen(a|c) &\geq \hmingen(a|b)+\hmingen(b|c) \label{eq:chainMinMinMin}\\
        \hmingen(a|c) &\leq \hmg(a|b) + \hmingen(b|c)\label{eq:chainMinMaxMin}\\
        \hmg(a|c) &\geq \hmg(a|b) + \hmingen(b|c)\label{eq:chainMaxMaxMin}~.
    \end{align}
\end{thm}
\begin{proof}
        These relations hold trivially for the area terms in Eq.~\eqref{eq:hmgdef}; and they hold by Ref.~\cite{Vitanov_2013} for the conditional max and min matter entropies. At higher orders in $G$, the chain rules would need to be stated as conjectures~\cite{Bousso:2024iry}. 
\end{proof}

\begin{thm}\label{thm:emaxaccessible} 
$\emax(a)$ is accessible from $a$.
\end{thm}

\begin{proof} 
We must show that $\emax(a)\in F(a)$, i.e., that $\emax(a)$ satisfies properties I--III listed in Def.~\ref{def:accessible}.

\emph{Property I:}
Since $a\subset \tilde a'$, $f=a$ satisfies properties I--III with any choice of Cauchy slice. Hence $F(a)$ is nonempty, and Eq.~\eqref{eq:def:emax} implies that $\tilde a'  \supset \emax(a) \supset a$. 

\emph{Property II:} Proceeding inductively, let $f_1,f_2\in F(a)$ and $f_3=f_1\Cup f_2$. By Corollary~48 of Ref.~\cite{Bousso:2024iry}, $f_3$ is antinormal at points $p \in \eth f_3 \setminus \eth a$. By induction, $\emax(a)$ is antinormal at points $p \in \eth \emax(a) \setminus \eth a$.

\emph{Property III:} Again proceeding inductively, let $f_1,f_2\in F(a)$, with property III satisfied by Cauchy slices $\Sigma_1$ and $\Sigma_2$, respectively. Then $f_3=f_1\Cup f_2$ admits the Cauchy slice
\begin{equation}\label{eq:unionslice}
  \Sigma_3 = \Sigma_1 \cup  [H^+(f_1')\cap J^-(\Sigma_2)] 
  \cup [H^-(f_1')\cap J^+(\Sigma_2)] \cup [\Sigma_2 \cap f_1']~.
\end{equation}
Let $h\supset a$, $\eth h\subset \Sigma_3$. By property III of $\Sigma_1$ and $\Sigma_2$, 
\begin{align}
    \hmg[\Sigma_1|h\cap\Sigma_1] &\leq 0~, \\
    \hmg[\Sigma_2|h\cap\Sigma_2] & \leq 0~.
\end{align}
Then strong subadditivity of the generalized conditional max entropy (Lemma~\ref{lem:ssa}) implies
\begin{align}
    \hmg[h\Cup \Sigma_1|h] &\leq 0~, \\
    \hmg[\Sigma_3|h \Cup (\Sigma_3\setminus \Sigma_2)] &\leq 0~.
\end{align}
By Discrete Max-Focusing, Conj.~\ref{conj:qfc}, 
\begin{equation}
    \hmg[h \Cup (\Sigma_3\setminus \Sigma_2)|h\Cup \Sigma_1] \leq 0~.
\end{equation}
Finally, by the chain rule Eq.~\eqref{eq:chainMaxMaxMax},
\begin{align}
    \hmg[\Sigma_3|h] &\leq \hmg[\Sigma_3|h \Cup (\Sigma_3\setminus \Sigma_2)] \\
    &\qquad + \hmg[h \Cup (\Sigma_3\setminus \Sigma_2)|h\Cup \Sigma_1] + \hmg[h\Cup \Sigma_1|h] \nonumber\\
    &\leq 0\nonumber~.
\end{align}
\end{proof}

\begin{thm}\label{lem:eminpropshare}
$\emin(a)\in G(a)$.
\end{thm}

\begin{proof} 
We must show that $\emin(a)$ satisfies properties i--iii listed in Def.~\ref{def:emin}.

\emph{Property i:} Consider $g=\tilde a '$. Since $\tilde a\subset a'$ by definition, $g$
satisfies property i. By Lemma \ref{lem:ta_anormal},   
$\tilde a$ is antinormal and satisfies property ii. Since $g\cap \tilde a' =\varnothing$, property iii is trivially satisfied.
Thus, $G(a)$ is nonempty. Property i then implies $\tilde a'\supset \emin(a)\supset a$.

\emph{Property ii:} The union of antinormal wedges is antinormal by Theorem 46 in Ref.~\cite{Bousso:2024iry}. Hence $\emin(a)'=(\Cup g')'$ is antinormal by Eq.~\eqref{eq:emindef}.

\emph{Property iii:} Let $g_1,g_2\in G(a)$ with property iii satisfied by Cauchy slices $\Sigma'_1$ and $\Sigma'_2$, respectively; and let $g_3=g_1\cap g_2$. Then $g_3'$ admits the Cauchy slice
\begin{equation}\label{eq:sminunion}
  \Sigma'_3 = \Sigma'_1 \cup  [H^+(g_1)\cap J^-(\Sigma'_2)] 
  \cup [H^-(g_1)\cap J^+(\Sigma'_2)] \cup [\Sigma'_2 \cap g_1]~.
\end{equation}
Let $h\supset g_3$, $\eth h\subset \Sigma'_3$. By property iii of $\Sigma'_1$ and $\Sigma'_2$, 
\begin{align}
    \hmg[\Sigma_1'|h'\cap\Sigma_1'] & \leq 0~,\\
    \hmg[\Sigma'_2|h'\cap\Sigma'_2] & \leq 0~.
\end{align}
Then  strong subadditivity of the generalized conditional max entropy (Lemma~\ref{lem:ssa}) implies
\begin{align}
    \hmg[h'\Cup \Sigma_1'|h'] & \leq 0~, \\
    \hmg[\Sigma'_3|h' \Cup (\Sigma_3'\setminus \Sigma_2')] & \leq 0~.
\end{align}
By Discrete Max-Focusing (Conj.~\ref{conj:qfc}),
\begin{equation}
    \hmg[h' \Cup (\Sigma_3'\setminus \Sigma_2')|h'\Cup \Sigma_1'] \leq 0~.
\end{equation}
Summing the above three inequalities and using the chain rule Eq.~\eqref{eq:chainMaxMaxMax}, we obtain $\hmg[\Sigma'_3|h']\leq 0$.
Hence $g_3$ satisfies property iii.
\end{proof}

\begin{conv}
Let $\smin(a)$ denote a Cauchy slice of $\emin'(a)$ that satisfies property iii, and which exists by the preceding Theorem. 
\end{conv}

\begin{thm}\label{thm:emaxthroat}
$\emax(a)$ is a throat in the spacetime $\tilde a'$.
\end{thm}
\begin{proof} 
Suppose, for contradiction, that $\emax(a)$ is not FNC in $\tilde a'$. Then there exists a past outward deformation of $\emax(a)$ in $\tilde a'$, denoted $f_1$,  that is FNE. By discrete focusing, $f_1$ is also PNE. Hence, this deformation satisfies Properties~I and II of Definition~\ref{def:accessible}. 

$f_1$ also satisfies property III. To verify this, consider the Cauchy surface $\Sigma_1 = \smax \cup \left(H^-(\emax(a)') \cap f_1\right)$, where $\smax$ is the Cauchy surface for $\emax(a)$ on which Property~III holds. Therefore, $f_1 \in F(a)$, which contradicts the definition of $\emax$ in Definition~\ref{def:emax}.

Similarly, one can show that $\emax(a)$ must be PNC by considering a future outward deformation. Then, by  \ref{thm:emaxaccessible}, $\emax(a)$ is future- and past-marginal accessible from $a$.
\end{proof}

\begin{cor}\label{thm:eminmarginal}
$\emin(a)'$ is future- and/or past-marginal on $\eth \emin(a)\setminus\eth a$. Specifically,
\begin{itemize}
    \item $\emin(a)'$ is future-marginal on $\eth \emin(a) \cap H^+(a')$.
    \item $\emin(a)'$ is past-marginal on $\eth \emin(a) \cap   H^-(a')$.
    \item $\emin(a)'$ is a throat on $\eth \emin(a) \cap a'$.
\end{itemize}
\end{cor}
\begin{proof}
This was previously proven as Theorem 64 of Ref.~\cite{Bousso:2024iry}. Here we note that the preceding theorem and Lemma~\ref{tta_empty} imply that $\emax(\tilde a)|_{a'}$ is a throat in $a'$. The result then follows immediately from the complementarity theorem~\ref{thm:maxmincomp}, $\emin(a)' = \emax(\tilde a)|_{a'}$.
\end{proof}

\begin{thm}\label{thm:emaxemin}
$\emax(a)\subset\emin(a)$.
\end{thm}

\begin{proof}
    Since only one input wedge $a$ is involved, we suppress the argument of $\emax$ and $\emin$. Note that
    \begin{equation}
        ((\emax' \cap \smin)\Cup(\emax'\cap\emin))' = (\emax' \cap \smin)'\cap(\emax'\cap\emin)'~.
    \end{equation}
    By Discrete Max-Focusing (Conj.~\ref{conj:qfc}),
    \begin{align}
        \hmg[(\emax' \cap \smin)'\cap(\emax'\cap\emin)'|\emax] &\leq 0 \\
        \hmg[(\emin \cap \smax)'\cap(\emin\cap\emax')'|\emin'] &\leq 0~.
    \end{align}
    Then
    \begin{align*}
        0 &\geq\hmg[(\emax' \cap \smin)'\cap(\emax'\cap\emin)'|\emax] + \hmg[\emax|\emin \cap \smax] \\
        &\geq \hmg[(\emax' \cap \smin)'\cap(\emax'\cap\emin)'|\emin \cap \smax] \\
        &= -\hmingen[(\emin \cap \smax)'|((\emax' \cap \smin)'\cap(\emax'\cap\emin)')'] \\ 
        &\geq -\hmg[(\emin \cap \smax)'|(\emax' \cap \smin)\Cup(\emax'\cap\emin)]~.
    \end{align*}
   By  strong subadditivity of the generalized conditional max entropy (Lemma~\ref{lem:ssa}),
    \begin{align*}
        0&\geq -\hmg[(\emin \cap \smax)'\cap(\emin\cap\emax')'|\emax' \cap \smin] \\
        &\geq -\hmg[(\emin \cap \smax)'\cap(\emin\cap\emax')'|\emin'] -\hmg[\emin'|\emax' \cap \smin] \\
        &\geq 0~.
    \end{align*}
    This is a contradiction, unless $\emin \cap \smax = \emax$ and $(\emax' \cap \smin)'=\emin$. Hence $\emax\subset\emin$.
\end{proof}

\begin{thm}[Nesting of $\emin$]\label{thm:nesting}
For wedges $a$ and $b$,
\begin{equation}
    a\subset b \implies \emin(a)\subset \emin(b)~.
\end{equation}
Moreover, $\smin(a)$ can be chosen so that
\begin{equation}
    \smin(a)\supset\smin(b)~.
\end{equation}
\end{thm}
\begin{proof}
     The proof differs through several details from the analogous proofs in Refs.~\cite{Bousso:2023sya, Bousso:2024iry} due to our modification of properties III and iii in the definition of $\emax$ and $\emin$.

Consider $g=\emin(b) \cap\tilde a'$. Note that $g\subset \emin(b)$.  We will show that $\emin(a) \subset g$ by proving $g\in G(a)$. The union of two antinormal wedges is antinormal by Theorem 46 of Ref.~\cite{Bousso:2024iry},
so $g'$ is antinormal and satisfies property ii. By definition, $g\subset \tilde a'$. By assumption, we have $a \subset b$ and Theorem \ref{lem:eminpropshare} implies that $a\subset \emin(b)$. Since $a\subset \tilde a'$, $g$ satisfies property i. Consider the following Cauchy slice of $g'$
\begin{equation}
  \Sigma' = \smin(b) \cup  [H^+(\emin(b))\cap J^-(\Sigma_{\tilde a})] 
  \cup [H^-(\emin(b))\cap J^+(\Sigma_{\tilde a})] \cup [\Sigma_{\tilde a} \cap \emin(b)]~
\end{equation}
where $\Sigma_{\tilde a}$ is a Cauchy slice of $\tilde a$.
Let $\tilde a'  \supset h\supset g$ be a wedge such that $\eth h \subset \Sigma'$.
By property iii of $\Sigma'_{min}$,
\begin{equation}
    \hmg(\emin(b)'| \Sigma_{min}(b)' \cap h') \leq 0 ~.
\end{equation}
Strong subadditivity of the generalized conditional max entropy (Lemma~\ref{lem:ssa})
implies
\begin{equation}
   \hmg(h'\Cup \Sigma_{min}(b)' | h')\leq 0
\end{equation}
By Discrete Max-Focusing, Conjecture~\ref{conj:qfc}, we have 
\begin{equation}
    \hmg(h'\Cup (\Sigma'\setminus\Sigma_{\tilde a })| h')\leq 0
\end{equation}
Since $h\subset \tilde a'$ we have 
\begin{equation}
   h'\Cup (\Sigma'\setminus \Sigma_{\tilde a}) = \Sigma'
\end{equation}
Thus we have 
\begin{equation}
    \hmg(\Sigma'| h')\leq 0
\end{equation}
Thus $g$ satisfies property iii with Cauchy slice $\Sigma'$.  

Suppose now that $\smin(a)\not\supset \smin(b)$. Then set $\tilde\Sigma = \smin(a)$ and redefine
\begin{equation}
    \smin(a)\equiv \smin(b) \cup \left( H^+[\emin(b)]\cap J^-(\tilde\Sigma)\right)
    \cup \left( H^-[\emin(b)]\cap J^+(\tilde\Sigma)\right)
    \cup [\tilde\Sigma\cap \emin(b)]~.
\end{equation}
This satisfies property iii for $\emin(a)$.
Thus we have $\emin(a) \subset g\subset \emin(b)$. 
\end{proof}
\begin{cor}
    If $a\subset b$ and $\eth \emin(b)\setminus \emin(a)$ is compact, then 
    \begin{equation}
        \hmg[\emin(a)'|\emin(b)']\leq 0~.
    \end{equation}
\end{cor}

\begin{proof}
    The proof is identical to the proof of Corollary 28 in  Ref.~\cite{Bousso:2023sya}.
\end{proof}

\begin{thm}[No Cloning]\label{thm:nocloning}
   \begin{equation}
     a \subset \emin'(b)~\mbox{and}~b \subset \emax'(a)~\implies~
     \emax(a)\subset\emin'(b)~.
  \end{equation}
\end{thm}

\begin{proof}
Let 
\begin{equation}
    g=\emin(b)\cap \emax'(a)~.
\end{equation}
We will show that $g$ satisfies properties i-iii listed in Def.~\ref{def:emin}. This contradicts the definition of $\emin(b)$ unless $g=\emin(b)$, which is equivalent to the conclusion.

Property i: By assumption, $b\subset \emax'(a)$. By Theorem~\ref{lem:eminpropshare}, $b\subset \emin(b)\subset \tilde b'$. Hence $b\subset g\subset \tilde b'$.

Property ii: By Theorem~\ref{lem:eminpropshare}, $\emin(b)'$ is antinormal, and by Theorem~\ref{thm:emaxaccessible}, $\emax(a)$ is antinormal. By Theorem 46 of Ref.~\cite{Bousso:2024iry}, $g'$ is antinormal. 

Property iii: Let 
\begin{multline}
    \Sigma' = \smin(b) 
    \cup \left( H^+[\emin(b)] \cap J^-[\smax(a)]\right) \\
    \cup \left( H^-[\emin(b)] \cap J^+[\smax(a)]\right)
    \cup [\smax(a) \cap \emin(b)]~.
\end{multline}
This is a Cauchy slice of $g'$. Let $\tilde b \supset h\supset g$ with $\eth h\subset \Sigma'$. By property iii of $\smin(b)$ and property III of $\smax(a)$
\begin{align}
    \hmg[\smin(b)|h'\cap\smin(b)] & \leq 0~,\\
    \hmg[\smax(a)|h'\cap\smax] & \leq 0~.
\end{align}
Then by strong subadditivity of the generalized conditional max entropy (Lemma~\ref{lem:ssa})
,
\begin{align}
    \hmg[h'\Cup \smin(b)|h'] & \leq 0~,\\
    \hmg[\Sigma'|h' \Cup (\Sigma'\setminus \smax(a))] & \leq 0~.
\end{align}
By Discrete Max-Focusing (Conj.~\ref{conj:qfc}),
\begin{equation}
    \hmg[h' \Cup (\Sigma'\setminus \smax(a))|h'\Cup \smin] \leq 0~.
\end{equation}
Finally, by the chain rule Eq.~\eqref{eq:chainMaxMaxMax}, we see that
\begin{align*}
    \hmg[\Sigma'|h'] &\leq \hmg[h'\Cup \smin|h'] + \hmg[h' \Cup (\Sigma'\setminus \smax(a))|h'\Cup \smin] \\
    &\qquad + \hmg[\Sigma'|h' \Cup (\Sigma'\setminus \smax(a))] \\
    &\leq 0 ~.
\end{align*}
Hence $g$ satisfies property iii.
\end{proof}

\begin{thm}[Strong Subadditivity of the Generalized Max and Min Entropies of Entanglement Wedges]
\label{thm:ssa}
Suppressing $\Cup$ symbols where they are obvious, let $a$, $b$, and $c$ be mutually spacelike wedges, such that
\begin{align}\nonumber
    \emin(ab) &=\emax(ab)~,~\emin(bc)=\emax(bc)~,\\ \nonumber
    \emin(b) &=\emax(b)~,~\mbox{and}~\emin(abc)=\emax(abc)~.
\end{align}
Then (writing $e$ for $\emin=\emax$)
\begin{align}
    \hmg[e(bc)|e(b)] &\geq \hmg[e(abc)|e(ab)]~;\label{eq:maxssa}\\
    \hmingen[e(bc)|e(b)] &\geq \hmingen[e(abc)|e(ab)]\label{eq:minssa}~.
\end{align}
\end{thm}

\begin{proof}
Following~\cite{Bousso:2023sya}, we define a wedge $x$ by the Cauchy slice of its complement:
\begin{equation}\label{eq:xdef}
    \Sigma'(x) = \Sigma'[e(ab)] \cup
    \left( H^+[e(ab)] \cap J^-[\eth e(bc)]\right) \cup
    \left( H^-[e(ab)] \cap J^+[\eth e(bc)]\right)~.
\end{equation}
Note that $\eth x$ is nowhere to the past or future of $\eth e(bc)$. Therefore, there exists a single Cauchy slice that contains the edges of $x$, $e(bc)$, $x\cap e(bc)$, and $x\Cup e(bc)$. We note that $x\cap e(bc)=e(ab)\cap e(bc)$ and that $x\Cup e(bc) \subset e(ab)\Cup e(bc)$. By Theorem~\ref{thm:nesting},
\begin{equation}
    x\Cup e(bc) \subset e(abc)~~\mbox{and}~~x\cap e(bc) \supset e(b)~.
\end{equation}

By Theorem 46 of Ref.~\cite{Bousso:2024iry} and Lemma~\ref{lem:ta_anormal}, $x'\Cup e(bc)'\Cup \tilde b$ is antinormal. By Discrete Max-Focusing,
property iii of $e(b)$, and the chain rule Eq.~\eqref{eq:chainMinMinMin}, respectively,
\begin{align}
    &\hmingen[x\cap e(bc)\cap \tilde b'|e(b) \Cup (\Sigma'(b)\cap x\cap e(bc)\cap \tilde b')]\geq0 \\
    &\hmingen[e(b) \Cup (\Sigma'(b)\cap x\cap e(bc)\cap \tilde b')|e(b)]\geq 0 \\
    &\hmingen[x\cap e(bc)\cap \tilde b'|e(b)]\geq 0~, ~.
\end{align}
We obtain  $x\cap e(bc) \cap \tilde b'$  from  $x\cap e(bc)$ by \begin{itemize}
    \item Deforming $x\cap e(bc)$ along its future lightsheet wherever the edge of $x\cap e(bc) \cap \tilde b'$  lies on the $L^+(a)$. 
\item Similarly deforming $x\cap e(bc)$ along its past lightsheet.
\item On all other points on edge of $x\cap e(bc)$ that is not already on the edge of $x\cap e(bc) \cap \tilde b'$ we first deform  $x\cap e(bc)$ along $L^+(a)$ until we land on $H^+(\tilde b)$. We then deform along $H^+(\tilde b)$ towards the edge of $x\cap e(bc) \cap \tilde b'$. 
\end{itemize}
Since $x'\Cup e(bc)'$ is antinormal and $H(\tilde b)$ is a causal horizon, the Discrete Max-Focusing Conjecture~\ref{conj:qfc} and the generalized second law imply that 
\begin{equation}\label{eq:y_cap_tb}
    H_{min,gen}[x\cap e(bc) |x\cap e(bc)\cap \tilde b' ] \geq 0 ~. 
\end{equation}
By Eq.~\ref{eq:y_cap_tb} and the chain rule Eq.~\eqref{eq:chainMinMinMin} 
\begin{equation}\label{eq:SSAproof min chain}
    \hmingen[x\cap e(bc)|e(b)]\geq 0~.
\end{equation}

Although $x$ need not be antinormal, Eq.~\eqref{eq:xdef} ensures that $[\eth x\setminus \eth e(ab)]\cap\eth(x\Cup e(bc))=\varnothing$. Arguments analogous to the proof of Corollary 48 of Ref.~\cite{Bousso:2024iry} then imply that $x\Cup e(bc)$ is antinormal, except where its edge coincides with $\eth[e(abc)]$ and hence with $\Sigma[e(abc)]$. By Discrete Max-Focusing, property III of $e(abc)$, and the chain rule Eq.~\eqref{eq:chainMaxMaxMax},
\begin{align}
    &\hmg[\Sigma(abc)\setminus (x' \cap e(bc)')|x\Cup e(bc)]\leq0 \\
    &\hmg[e(abc)|\Sigma(abc)\setminus (x' \cap e(bc)')] \leq 0 \\
    &\hmg[e(abc)|x\Cup e(bc)]\leq 0~. \label{eq:SSAproof max chain}
\end{align}
Normality of $e(ab)$ and the QFC imply $\hmg[x'|e(ab)']\leq 0$, i.e.,
\begin{equation}\label{eq:SSAproof qfc}
    \hmingen[e(ab)|x]\geq 0~.
\end{equation}

We now prove Eq.~\eqref{eq:maxssa} as follows:
\begin{align}
    \hmg[e(bc)|e(b)] &\geq \hmg[e(bc)|x\cap e(bc)] + \hmingen[x\cap e(bc)|e(b)] \\
    & \geq \hmg[e(bc)|x\cap e(bc)] \\
    &\geq \hmg[x\Cup e(bc)|x] \\
    &\geq \hmg[e(abc)|x\Cup e(bc)] + \hmg[x\Cup e(bc)|x] \\
    &\geq \hmg[e(abc)|x] \\
    &\geq \hmg[e(abc)|e(ab)] + \hmingen[e(ab)|x] \\
    &\geq \hmg[e(abc)|e(ab)]~,
\end{align}
where we applied the chain rule Eq.~\eqref{eq:chainMaxMaxMin} in the first inequality; Eq.~\eqref{eq:SSAproof min chain} in the second; strong subadditivity of the generalized max entropy in the third; Eq.~\eqref{eq:SSAproof max chain} in the fourth; the chain rule Eq.~\eqref{eq:chainMaxMaxMax} in the fifth; the chain rule Eq.~\eqref{eq:chainMaxMaxMin} in the sixth, and Eq.~\eqref{eq:SSAproof qfc} in the final inequality.

The proof of Eq.~\eqref{eq:minssa} proceeds similarly:
\begin{align}
    \hmingen[e(abc)|e(ab)] &\leq \hmingen[e(abc)|e(ab)] + \hmingen[e(ab)|x]\\
    &\leq \hmingen[e(abc)|x]\\
    &\leq \hmg[e(abc)|x\Cup e(bc)] + \hmingen[x\Cup e(bc)|x] \\
    &\leq \hmingen[x\Cup e(bc)|x] \\
    &\leq \hmingen[e(bc)|x\cap e(bc)] \\
    &\leq \hmingen[e(bc)|x\cap e(bc)] + \hmingen[x\cap e(bc)|e(b)] \\
    &\leq \hmingen[e(bc)|e(b)]~,
\end{align}
where we applied Eq.~\eqref{eq:SSAproof qfc} in the first inequality; the chain rule Eq.~\eqref{eq:chainMinMinMin} in the second; the chain rule Eq.~\eqref{eq:chainMinMaxMin} in the third; Eq.~\eqref{eq:SSAproof max chain} in the fourth; strong subadditivity of the generalized min entropy,  in the fifth; Eq.~\eqref{eq:SSAproof min chain} in the sixth; and the chain rule Eq.~\eqref{eq:chainMinMinMin} in the final inequality.
\end{proof}
\begin{cor}
    One can eliminate the assumption that $\emin(b)=\emax(b)$, so long as $e(b)$ is replaced by $\emin(b)$ in the conclusions of the theorem. 
\end{cor}

\section{Large Fundamental Complement}
\label{app:largecomp}

In this Appendix we present a possible alternative definition of the fundamental complement. We will denote this alternative by $\hat a$, and we will it the \emph{large fundamental complement} because $\hat a \supset \tilde a$ (see Lemma~\ref{lem:ta_subset_ha}). The property $\hat a\subset a'$ will still hold.

With $\hat a$ replacing $\tilde a$ in Def.~\ref{def:accessible} of accessibility, and thus in the definitions of $\emin$ and $\emax$, entanglement wedges will generically become ``smaller.'' But all theorems, including complementarity, remain valid, as can be verified by substituting $\hat a$ for $\tilde a$ in their proofs.

However, the large fundamental complement appears less motivated to us on physical grounds, because $\hat a$ need not contain timelike curves that are infinite in both the future and past direction. It is this property --- that $\tilde a$ is a causal patch of a subset of conformal infinity --- that makes $\tilde a$ quite obviously uncontrollable by $a$ and justifies its exclusion from $M$ when one constructs $\emax(a)$ and $\emin(a)$.

\subsection{Definition} \label{sec:largedef}

\begin{defn} (Conformal Shadow) \label{def:conf_shadow}
The conformal shadow $\tilde \scri$ of a wedge $a$ is the subset of $\scri$ that is causally inaccessible from $a$: 
\begin{equation}
    \tilde \scri(a) \equiv  \scri \setminus [\cl I(a)]_{\bar M}~ .
\end{equation}
The subscript indicates that $\cl I(a)$ should be computed in $\bar M$. Equivalently,
\begin{equation}\label{eq:tsae}
    \tilde \scri(a) = \setint_{\scri} (\scri \cap [\cl (a')]_{\bar M}) ~.
\end{equation}
\end{defn}

\begin{defn} (Large Fundamental Complement) \label{def:largefundcomp}
The large fundamental complement $\hat a$ of a wedge $a$ is the double complement of the conformal shadow of $a$. More precisely,\footnote{If an external nongravitational quantum system $R$ is invoked, for example in order to purify the overall quantum state, then $R$ should be included with $\hat a$ if and only if $R$ is not included with $a$.}
\begin{equation}\label{eq:tildea}
    \hat a \equiv  (\tilde \scri(a)'_{\bar M}\cap M)'~.
\end{equation}
\end{defn}

\subsection{Properties} \label{sec:largeprop}

\begin{lem}\label{lem:pq}
    Let $s=\{p,q\}$, where $q\in I^-(p)$. Then $s'$ is a wedge. 
\end{lem}
\begin{proof}
    Since $p$ is in the chronological future of $q$, the null hypersurfaces $\partial I^-(p)$ and $\partial I^+(q)$ intersect on a codimension-2 spacelike surface $\gamma$. Let $\Sigma$ be a Cauchy slice of $M$ to the past of both $p$ and $q$. Then $\gamma$ splits the Cauchy slice $\Sigma' = [\partial I^-(p)\cap J^+(\Sigma)] \cup [\Sigma\setminus I^-(p)]$ into the two partial Cauchy slices $\Sigma_{\rm in}=\partial I^-(p)\cap I^+(q)$ and $\Sigma_{\rm out}=\Sigma'\setminus \Sigma_{\rm in}\setminus \gamma$. The domain of dependence of any partial Cauchy slice is a wedge~\cite{Bousso:2023sya}. Therefore $s'=M \setminus \cl I^-(p) \setminus \cl I^+(q) = M\setminus \cl I(\Sigma_{\rm in}) =  D[\Sigma_{\rm out}]$ is a wedge. (See Sec.~III in Ref.~\cite{Bousso:2015qqa} for a detailed discussion of the final equality.)
\end{proof}

\begin{cor}
    Let $q\in I^-(p)$ and let $\eta$ be a causal curve from $p$ to $q$. Then $\eta'$ is a wedge.
\end{cor}

\begin{proof}
    Since $I^-(\eta)\cap I^+(\eta) = I^-(p)\cap I^+(q)$, the result follows immediately from Lemma~\ref{lem:pq}.
\end{proof}

\begin{cor}\label{cor:comp_open}
If $s$ is an open set in $M$, then $s'$ is a wedge.
\end{cor}
\begin{proof}
    Let $p\in s$. Since $s$ is open, there exists a point $q_p\in s\cap I^-(p)$. Hence
    \begin{equation}
        s' = \cap_{p\in s} [M \setminus \cl I(p)] = \cap_{p\in s} [M \setminus \cl I(p) \setminus \cl I(q_p)] = \cap_{p\in s}(\{p,q_p\}')~.
    \end{equation}
    By Lemma~\ref{lem:pq}, $\{p,q\}'$ is a wedge. Since an intersection of wedges is a wedge~\cite{Bousso:2023sya}, $s'$ is a wedge.
\end{proof}

\begin{lem}
    $\hat a$ is a wedge in $M$.
\end{lem}
\begin{proof}
   $ \tilde \scri(a)'_{\bar M}$ and M are open in $\bar M$, so their intersection is an open subset of $M$. By Corollary \ref{cor:comp_open}, $\hat a$ is a wedge.
\end{proof}

\begin{lem}\label{lem:large_ta_in_ap}
For any wedge $a$, the large fundamental complement $\hat a$ is a subset of the complement $a'$:
\begin{equation}
    \hat a\subset a' ~.
\end{equation}
\end{lem}
\begin{proof}
    The result follows immediately from Eq.~\eqref{eq:tsae}.
\end{proof}
\begin{lem}\label{lem:ta_subset_ha}
For any wedge $a$ the fundamental complement is a subset of the large fundamental complement:
\begin{equation}
    \tilde a \subset \hat a ~.
\end{equation}
\end{lem}
\begin{proof}
It follows from the definition of $\tilde \scri (a)$ that \begin{equation}
    \tilde \scri (a) =\scri|_{a'} \supset \scri|_{\tilde a} ~.
\end{equation}
This implies $\hat a \supset \tilde a$.
\end{proof}
\begin{lem}\label{lem:large_ta_anormal}
    For any wedge $a$, the large fundamental complement $\hat a$ is antinormal.
\end{lem}
\begin{proof}
    The definition of $\hat a$ implies that 
   $ \eth \hat a$ lies on the intersection of a future and past causal horizon. (A future causal horizon is the boundary of the past of any subset of $\scri$~\cite{Bousso:2025xyc}.) By the Generalized Second Law of thermodynamics~\cite{Bekenstein:1972tm, Wall:2010jtc, Bousso:2025xyc}, $a$ is past and future nonexpanding.
\end{proof}

\begin{thm}\label{thm:large_ta_empty}
    If $\hat a= \varnothing$, then $\emax(a)=\emin(a)=M$.
\end{thm}

\begin{proof}
    The proof is the same as for Theorem~\ref{thm:ta_empty}.
\end{proof}

\begin{lem} \label{large_tta_empty}
      Let $\hat a$ be the large fundamental complement of $a$ in $M$, and let $\left. \widehat{(\hat a)}\right|_{a'}$ be the fundamental complement of $\hat a$ in $a'$. Then 
      \begin{equation}
          \left. \widehat{(\hat a)}\right|_{a'} = \varnothing. 
      \end{equation}
\end{lem}
\begin{proof}
By Eq.~\eqref{eq:tsae}, null infinity of the spacetime $a'$ is identical to the conformal shadow of $a$: $\scri|_{a'}= \scri\cap \cl  (a') = \tilde\scri(a)$. Then we have $\tilde \scri(a) \supset \tilde\scri(\hat a)|_{a'} $. This implies that $\hat a \supset  \left. \widehat{(\hat a)}\right|_{a'} $. By Lemma \ref{lem:large_ta_in_ap}, we also have $\hat a'|_{a'} \supset  \left. \widehat{(\hat a)}\right|_{a'}$ where $\hat a'|_{a'}$ is the complement of $\hat a$ in $a'$. This can only be true if $\left. \widehat{(\hat a)}\right|_{a'} = \varnothing$.
\end{proof}

\subsection{Differences Between the Large and Regular Fundamental Complement}\label{sec-largediff}

A key difference between $\tilde a$ and $\hat a$ is that the latter need not contain timelike curves that are both past- \emph{and} future-infinite; \emph{or} is enough. There are two important classes of spacetimes where this makes a difference for $\emax$ and $\emin$: universes with a Big Bang, or with a Big Crunch (but not both) need not be trivially reconstructible (though some are anyway). And there is now a larger class of input wedges from which global de Sitter space is \emph{not} trivially reconstructible. Examples are shown in Fig.~\ref{fig:ds-examples-large_fc} -\ref{fig:cosmology_examples}. 

\begin{figure}[H]
\centering
    \includegraphics[width=1\linewidth]{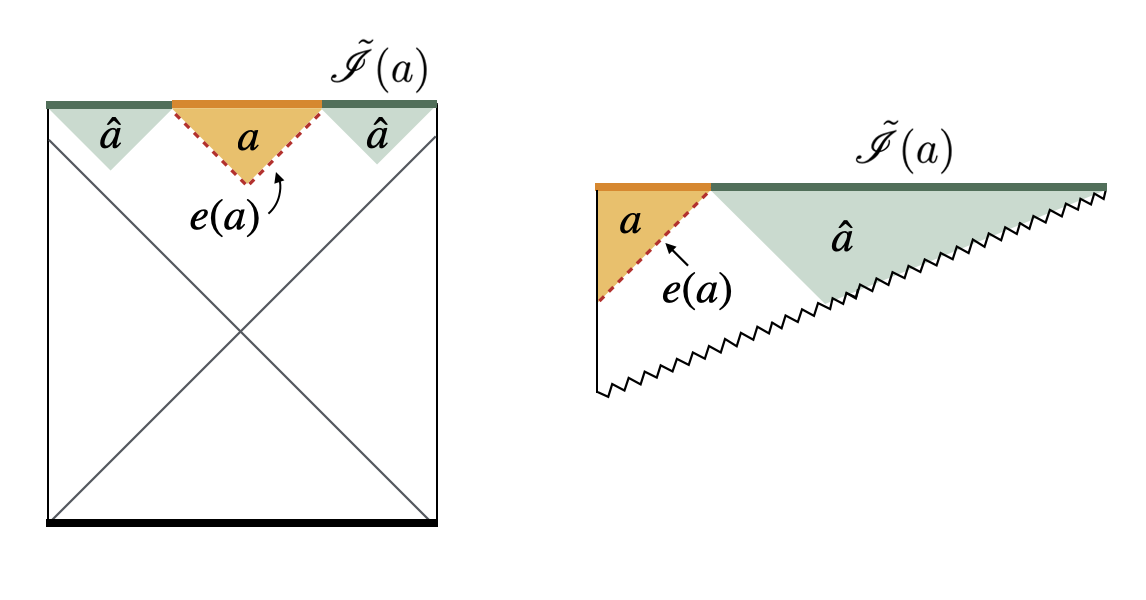}
    \caption{Examples in asymptotically de Sitter spacetimes where $\hat a=a'\neq\varnothing$ whereas $\tilde a=\varnothing$. The conformal shadow of $a$ is shown (thick green lines). The right figure illustrates that a Big Bang cosmology can have nonvanishing $\hat a$, and hence is not trivially reconstructible.}
    \label{fig:ds-examples-large_fc}
\end{figure}

\begin{figure}[H]
    \centering
    \includegraphics[width=1\linewidth]{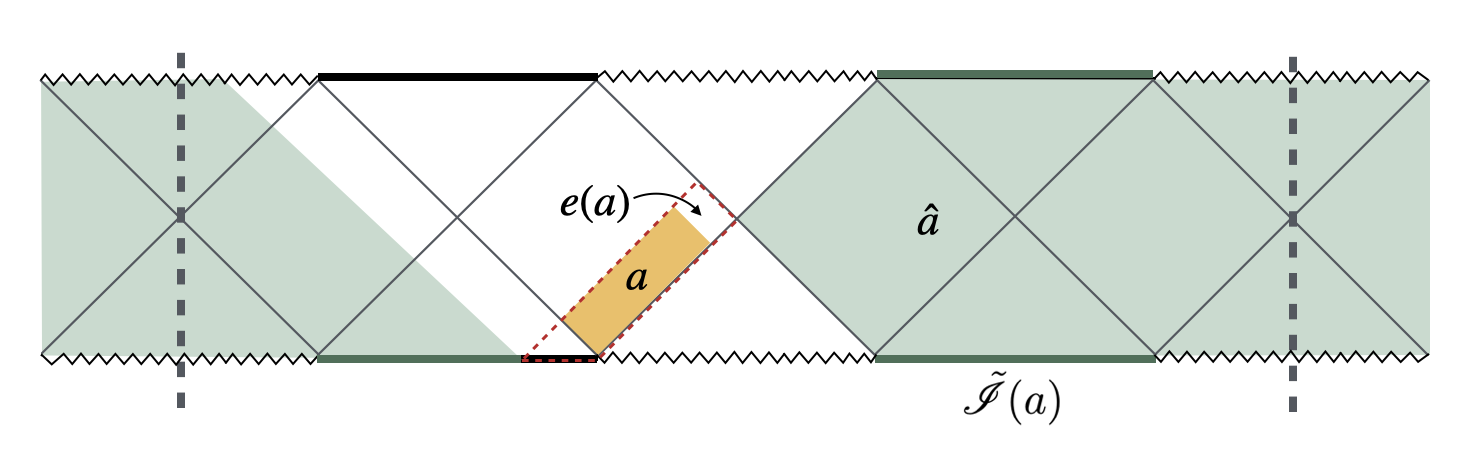}
    \caption{Schwarzschild-de Sitter spacetime; the wedge $a$ lies between the black hole and cosmological horizons. Comparison with Fig.~\ref{fig:dS-sch-ds} shows that $\hat a$ is larger than $\tilde a$, resulting in a smaller $e(a)$.}
    \label{fig:dS-sch-ds-large_fc}
\end{figure}

The analogue of Corollary~\ref{cor:trivial} is now the following statement: suppose that the spacetime $M$ contains no future-infinite timelike curve, or that $M$ contains no past-infinite timelike curve. (Loosely speaking this means that $M$ starts with a Big Bang, or ends with a Big Crunch.) Let $a\subset M$ be a wedge. Then $\emax(a)=\emin(a)=M$. Examples are shown Fig.~\ref{fig:cosmology_examples}. 

 \begin{figure}[H]
     \centering
     \includegraphics[width=1\linewidth]{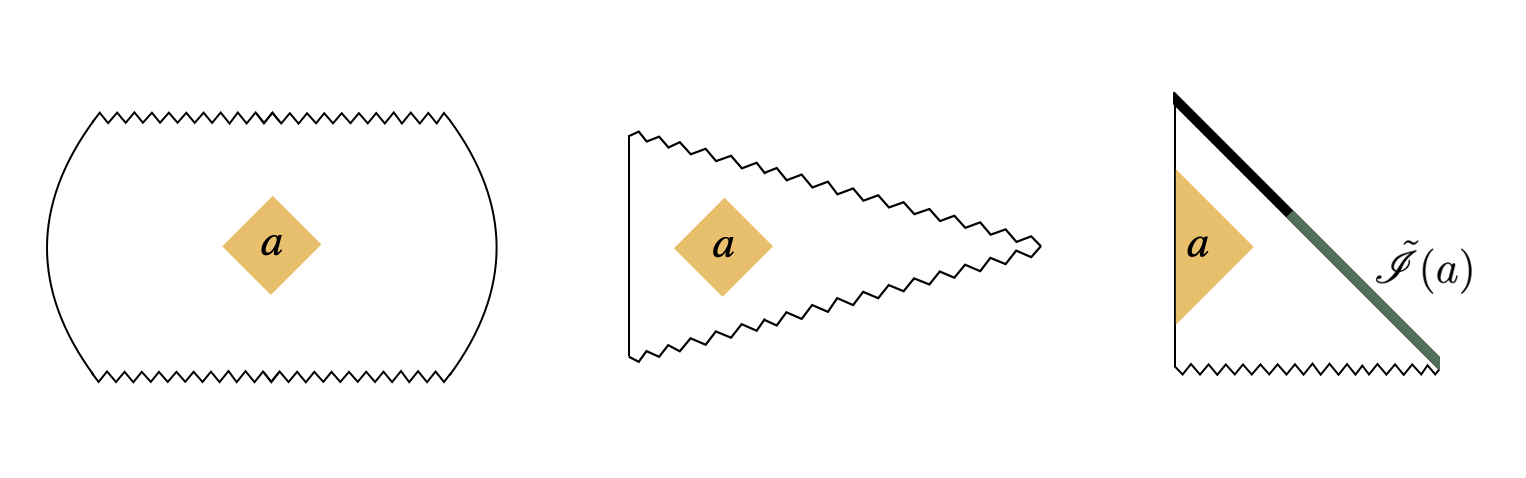}
     \caption{
    A sufficient condition for $\hat a=\varnothing$, and hence for $e(a)=M$, is the presence of both a Big Bang and a Big Crunch singularity. (Recall that for $\tilde a=\varnothing$, one of the two suffices.) This does not require the universe to be spatially closed (left); for example, an open Friedmann-Robertson-Lema\^itre-Walker cosmology with negative cosmological constant also recollapses (middle). Nor is the presence of both singularities a necessary condition for $\hat a=\varnothing$. The figure on the right shows an example of a spatially flat expanding universe with $\hat a=\varnothing$, $e(a)=M$.}
     \label{fig:cosmology_examples}
 \end{figure}

\bibliographystyle{JHEP}
\bibliography{fundamental_comp}
\end{document}